\documentclass[11pt,a4paper]{article}

\usepackage{fancyhdr}
\usepackage{amsmath,mathtools}
\usepackage{bm}
\usepackage{esint}
\usepackage{amsfonts}
\usepackage{amsthm}
\usepackage{amssymb}
\usepackage{amsbsy}
\usepackage{verbatim}
\usepackage{graphicx}
\usepackage{url}
\usepackage{mathrsfs}
\usepackage[hmargin = 1in,vmargin=.75in]{geometry}
\usepackage{color}
\usepackage{amstext}
\usepackage{stmaryrd}
\usepackage{enumerate}
\usepackage{algpseudocode}
\usepackage{algorithm}
\usepackage{pifont}
\usepackage{prettyref}
\usepackage{subcaption}
\usepackage{hyperref}

\font\msbm=msbm10

\numberwithin{equation}{section}

\theoremstyle{plain}
\newtheorem{theorem}{Theorem}[section]
\newtheorem{lemma}[theorem]{Lemma}
\newtheorem{corollary}[theorem]{Corollary}

\newtheorem{proposition}[theorem]{Proposition}

\def\mathbb#1{\hbox{\msbm{#1}}}

\newcommand{\ba}{\boldsymbol{a}}

\newcommand{\be}{\boldsymbol{e}}

\newcommand{\bx}{\boldsymbol{x}}

\newcommand{\bz}{\boldsymbol{z}}
\newcommand{\bone}{\boldsymbol{1}}

\newcommand{\bvarphi}{\boldsymbol{\varphi}}

\newcommand{\BA}{\boldsymbol{A}}
\newcommand{\BB}{\boldsymbol{B}}
\newcommand{\BC}{\boldsymbol{C}}
\newcommand{\BD}{\boldsymbol{D}}

\newcommand{\BJ}{\boldsymbol{J}}
\newcommand{\BL}{\boldsymbol{L}}

\newcommand{\BO}{\boldsymbol{O}}

\newcommand{\BQ}{\boldsymbol{Q}}
\newcommand{\BR}{\boldsymbol{R}}
\newcommand{\BS}{\boldsymbol{S}}

\newcommand{\BU}{\boldsymbol{U}}
\newcommand{\BV}{\boldsymbol{V}}
\newcommand{\BW}{\boldsymbol{W}}
\newcommand{\BX}{\boldsymbol{X}}
\newcommand{\BY}{\boldsymbol{Y}}
\newcommand{\BZ}{\boldsymbol{Z}}

\newcommand{\BDelta}{\boldsymbol{\Delta}}

\newcommand{\BLambda}{\boldsymbol{\Lambda}}
\newcommand{\BSigma}{\boldsymbol{\Sigma}}

\newcommand{\PP}{\mathcal{P}}

\newcommand{\I}{\boldsymbol{I}}
\newcommand{\RR}{\mathbb{R}}
\newcommand{\lag}{\langle}
\newcommand{\rag}{\rangle}

\newcommand{\eps}{\epsilon}

\DeclareMathOperator{\Tr}{Tr}

\DeclareMathOperator{\E}{\mathbb{E}}

\DeclareMathOperator{\diag}{diag}

\DeclareMathOperator{\SNR}{SNR}

\DeclareMathOperator{\St}{St}

\DeclareMathOperator{\Od}{O}

\DeclareMathOperator{\blkdiag}{blkdiag}
\DeclareMathOperator{\argmin}{argmin}

\newcommand\keywords[1]{\textbf{Keywords}: #1}
\newcommand\MSC[1]{\textbf{MSC numbers}: #1}

\begin{document}
\title{\bf Generalized Orthogonal Procrustes Problem \\ under Arbitrary Adversaries}
\author{Shuyang Ling\thanks{Shanghai Frontiers Science Center of Artificial Intelligence and Deep Learning, New York University Shanghai. Address: 567 Yangsi Road, Pudong New District, Shanghai, China, 200124. S.L. is (partially) financially supported by the National Key R\&D Program of China, Project Number 2021YFA1002800, Shanghai STCSM Rising Star Program No. 24QA2706200, STCSM General Program No. 24ZR1455300, National Natural Science Foundation of China No.12001372, Shanghai SMEC No. 0920000112, and NYU Shanghai Boost Fund.}}

\maketitle

\begin{abstract}
The generalized orthogonal Procrustes problem (GOPP) plays a fundamental role in several scientific disciplines including statistics, imaging science and computer vision. Despite its tremendous practical importance, it is generally an NP-hard problem to find the least squares estimator. We study the semidefinite relaxation (SDR) and an iterative method named generalized power method (GPM)  to find the least squares estimator, and investigate the performance under a signal-plus-noise model. We show that the SDR recovers the least squares estimator exactly and moreover the generalized power method with a proper initialization converges linearly to the global minimizer to the SDR, provided that the signal-to-noise ratio is large. The main technique follows from showing the nonlinear mapping involved in the GPM is essentially a local contraction mapping and then applying the well-known Banach fixed-point theorem finishes the proof.
In addition, we analyze the low-rank factorization algorithm and show the corresponding optimization landscape is free of spurious local minimizers under nearly identical conditions that enables the success of SDR approach. The highlight of our work is that the theoretical guarantees are purely algebraic and do not assume any statistical priors of the additive adversaries, and thus it applies to various interesting settings. 
\end{abstract}

\vskip0.4cm
\MSC{90C22, 65D18, 49M20, 65K05}

\section{Introduction}
Given a set of multiple point clouds, how to find a rotation for each point cloud such that all the transformed point clouds are well aligned? This problem, known as rigid point clouds registration and also generalized orthogonal Procrustes problem (GOPP), has found numerous applications in computer vision~\cite{OVBS17} (multi-view point clouds registration), statistics~\cite{G75,GD04,S66}, shape analysis~\cite{SG02},  cryo-electron microscopy~\cite{CKS15,PSB21,SS11,S18}, and robotics~\cite{RCBL19}. In particular, it is also an important subproblem of the Procrustes matching problem~\cite{BM92,DL17,KK16,MDKK16}.

In this work, we consider each point cloud $\BA_i$ is a noisy copy of a common point cloud $\BA\in\RR^{d\times m}$ transformed by an unknown $d\times d$ orthogonal matrix $\BO_i$:
\begin{equation}\label{def:model0}
\BA_i = \BO_i \BA + \BDelta_i, ~1\leq i\leq n,
\end{equation}
where $\{\BDelta_i\}_{i=1}^n$ is any arbitrary additive noise. Our goal is to recover $(\BA,\BO_i)$ from $\{\BA_i\}_{i=1}^n$. In matrix form, we define
\begin{equation}\label{def:model}
\BD := \begin{bmatrix}
\BA_1 \\
\vdots \\
\BA_n
\end{bmatrix} = \BO\BA + \BDelta \quad\text{where}\quad \BO := 
\begin{bmatrix}
\BO_1 \\
\vdots \\
\BO_n
\end{bmatrix}\in\Od(d)^{\otimes n}, \quad \BDelta:=
\begin{bmatrix}
\BDelta_1 \\
\vdots \\
\BDelta_n
\end{bmatrix} \in\RR^{nd\times m},
\end{equation}
where $\Od(d)$ denotes the set of all $d\times d$ orthogonal matrices.

\subsection{Least squares estimation}
One approach is to find the least squares estimator by minimizing $f(\BS,\BA):$
\begin{align*}
f(\BS, \BA) & := \sum_{i=1}^n \| \BS_i\BA  - \BA_i\|_F^2 = {n\|\BA\|^2_F} - 2\sum_{i=1}^n\lag \BA, \BS_i^{\top}\BA_i\rag + \sum_{i=1}^n \|\BA_i\|_F^2
\end{align*}
where $\BS^{\top} = [\BS_1^{\top},\cdots,\BS^{\top}_n]$ is an element in $\Od(d)^{\otimes n}$.
For any fixed $\BS\in\Od(d)^{\otimes n}$, the global minimizer for $\BA$ is 
$\widehat{\BA} = n^{-1}\sum_{i=1}^n \BS_i^{\top}\BA_i$. 
Therefore, the least squares estimation is equivalent to find the global minimizer to the following program:
\begin{equation}\label{def:od}
\min_{\BS_i\in\Od(d)}~- \sum_{i,j} \left\lag\BA_i\BA_j^{\top}, \BS_i\BS_j^{\top}\right\rag = \min_{\BS\in\Od(d)^{\otimes n}}~-\left\lag \BC,\BS\BS^{\top}\right\rag  \tag{P}
\end{equation}
where $\BC_{ij} = \BA_i\BA_j^{\top}$ and $\BC = \BD\BD^{\top}$  is the data matrix:
\begin{equation}\label{def:C}
\BC = (\BO\BA + \BDelta)(\BO\BA + \BDelta)^{\top} = \BO\BA\BA^{\top}\BO^{\top} + \BDelta\BA^{\top}\BO^{\top} + \BO\BA\BDelta^{\top} + \BDelta\BDelta^{\top}.
\end{equation}

For $n=2$,~\eqref{def:od} can be easily solved by performing singular value decomposition~\cite{S66}.
However, it becomes much more challenging for large $n$ to find the global minimizer to~\eqref{def:od}~\cite{G75,N07,S11b}. In fact,~\eqref{def:od} is a well-known NP-hard problem even if $d=1$. Therefore, one would seek for alternative methods to approximate the global solution to this nonconvex program.

\subsection{Convex relaxation and its tightness} Among many existing approaches, semidefinite relaxation (SDR) is one of the most powerful methods. 
One common SDR uses the fact that {$\BZ = \BS\BS^{\top}$} is positive semidefinite and also its diagonal block is $\I_d$. Using these two observations gives a convex program:{
\begin{equation}\label{def:sdr}
\min_{\BZ\in\RR^{nd\times nd}}~-\lag \BC, \BZ\rag \quad \text{s.t.}\quad \BZ\succeq 0,~ \BZ_{ii} = \I_d. \tag{SDR}
\end{equation}}
Despite that this SDR is solvable with a polynomial-time algorithm, the main question is whether this SDR produces a solution close to the least squares estimator. 
In~\cite{BKS14},  Bandeira, Khoo, and Singer conjectured that a threshold exists such that if the signal-to-noise ratio (SNR) is above this threshold, the SDR succeeds in recovering the least squares estimator. Roughly speaking, if the noise strength $\|\BDelta\|$ is weaker than the signal $\|\BA\|$, the SDR can find the global minimizer to a highly nonconvex problem. In this work, we will try to investigate the following question 
\begin{equation*}
\text{\em When does the SDR recover the least squares estimator?}
\end{equation*}
This is known as the~\emph{tightness} of the SDR: the global minimizer to~\eqref{def:sdr} is rank-$d$.

Convex relaxation has proven to be one powerful approach to solve the generalized orthogonal Procrustes problem~\cite{B15,BKS16,CKS15,NRV13,N07,SS11,S11b}. The~\eqref{def:sdr} is a convex relaxation of~\eqref{def:od} and it generalizes the famous Goemans-Williamson relaxation~\cite{GW95} of the graph Max-Cut. After solving the SDR, one needs to round the solution to orthogonal matrices. This relaxation plus rounding strategy has enjoyed remarkable successes and can produce satisfactory approximate solutions~\cite{BKS16,NRV13}. {On the other hand, as observed in~\cite{BKS14}, the SDR is tight in the numerical trials, i.e., it
directly produces the globally optimal solution to~\eqref{def:od} and the rounding procedure is not necessary.} Similar phenomena have been observed in a series of applications arising from signal processing and machine learning such as matrix completion~\cite{CR09}, blind deconvolution~\cite{ARR13}, phase retrieval~\cite{CESV15}, multi-reference alignment~\cite{BCSZ14}, and community detection~\cite{ABBS14} in which the optimal solution to the SDR is exactly the ground truth solution. However, unlike these aforementioned examples, the least squares estimator usually does not match the ground truth signal in the generalized orthogonal Procrustes problem {in presence of noise.} The similar issue also arises in angular synchronization~\cite{BBS17,FKM21,ZB18}, orthogonal group synchronization~\cite{L20a,L20c,SS11,WS13}, and orthogonal trace-sum maximization~\cite{WZZ22}. This has created many difficulties in analyzing the tightness of convex relaxation which requires significantly different techniques from the previous examples~\cite{ABBS14,ARR13,BCSZ14,CESV15,CR09}. { It is also worth noting that spectral, least squares and SDP estimator achieve min-max optimal for the group synchronization under Gaussian noise even though they are not equal to the ground truth~\cite{GZ22,GZ23,Z22}. }

\subsection{Generalized power method (GPM)}
As the SDR is expensive to solve, we also consider an iterative method named generalized power method to efficiently find the least squares estimator. The algorithm follows a two-step procedure:
(a) start with an initialization $\BO^0$ via spectral method    {(spectral initialization), which is shown in the first two steps in Algorithm 1}; (b) update the iterate $\BO^t$ by
\begin{equation}\label{def:gpm}
\BO^{t+1} = \PP_n(\BC\BO^t)
\end{equation}
where $\BC\BO^t\in\RR^{nd\times d}$ and $\PP_n(\cdot):\RR^{nd\times d}\rightarrow \Od(d)^{\otimes n}$ maps each $d\times d$ block to a $d\times d$ orthogonal matrix by
\begin{equation}\label{def:proj}
\PP(\BR) := \argmin_{\BQ\in\Od(d)}\|\BQ - \BR\|_F^2
\end{equation}
whose global minimizer is $\BQ = \BU\BV^{\top}$ where $\BU$ and $\BV$ are the left and right singular vectors of $\BR\in\RR^{d\times d}$ respectively. 

\begin{algorithm}[h!]
\caption{{Spectral initialization and }generalized power methods for the GOPP}
\begin{algorithmic}[1]\label{algo:gopp}
\State Compute the top $d$ left singular vectors $\BU\in\RR^{nd\times d}$ of $\BD$ in~\eqref{def:model} with $\BU^{\top}\BU= \I_d.$
\State Compute $\PP(\BU_i)$ for all $1\leq i\leq n$ where $\BU_i$ is the $i$th $d\times d$ block of $\BU$.
\State Initialize $\BO^{0} = \PP_n(\BU)\in\RR^{nd\times d}.$  

\State $\BO^{t+1} = \PP_n( \BC\BO^{t})$, \quad $t=0,1,\cdots$

\State Stop when the iteration stabilizes.

\end{algorithmic}
\label{algo1}
\end{algorithm}\label{alg}

Generalized power method is much more efficient than the convex relaxation: {Every iteration of the GPM is essentially a matrix-vector multiplication followed by a projection step for each block. As shown in~\cite{L20a,LYS17,LYS20,ZB18}, the GPM is significantly faster than the SDR in synchronization problem.} However, this two-step approach is essentially nonconvex and has the risk of getting stuck at stationary points or local minimizer. 
On the other hand, the output of the GPM is usually quite satisfactory and is even equal to the solution to~\eqref{def:sdr}. Therefore, we would like to theoretically justify why this efficient nonconvex approach is successful, especially in the regime of high signal-to-noise ratio (SNR). 
In particular, we are interested in answering this question:
\begin{equation*}
\text{\em When does the generalized power method to recover the least squares estimator exactly?}
\end{equation*}


In practice, the first-order optimization methods~\cite{AMS09,B15,EAS98,MGPG04,WY13} are quite popular in solving medium/large-scale convex programs. The main idea is first to find a proper initialization and show that global convergence of the algorithm is guaranteed with this carefully-chosen initialization~\cite{CLS15,CC18,KMO10,MWCC20}.
For the generalized orthogonal Procrustes problem, we will focus on the generalized power method (GPM) since it has been successfully applied to the group synchronization problem. The GPM was first proposed and analyzed in joint alignment problem~\cite{CC18}, angular synchronization~\cite{B16,LYS17,ZB18}, and was later extended to orthogonal group synchronization~\cite{LYS20,L20c} and community detection under the stochastic block model~\cite{WLZS21}.  The work~\cite{ZB18} by Zhong and Boumal analyzed the global convergence of the GPM to the least squares estimator in angular synchronization. Later on,~\cite{LYS20} studied the convergence of the GPM in the orthogonal group (and its subgroup) synchronization and provided a unified analysis on the error bound for the iterates. The analysis of the GPM algorithm combined with spectral initialization in~\cite{ZB18} has also been extended to the orthogonal group synchronization and the generalized orthogonal Procrustes problem with additive Gaussian noise in~\cite{L20b,L20c} and~\cite{L23d} respectively. However, it remains unclear whether the GPM works for the orthogonal Procrustes problem under arbitrary adversaries, which is the focus of this work.

\subsection{Optimization landscape of low-rank factorization}

The generalized power method works well empirically to solve the generalized Procrustes problem if the noise is small compared with the signal. More surprisingly, even without spectral initialization, the power method still recovers a tight solution. To partially explain this phenomenon, we investigate the performance of the Burer-Monteiro approach. 
The Burer-Monteiro approach~\cite{BM03,BM05,B15} replaces the orthogonality constraints of each $\BO_i$ in~\eqref{def:od} by an element on the Stiefel manifold $\St(d,p) : = \{\BS_i\in\RR^{d\times p}: \BS_i\BS_i^{\top} = \I_d\}$~\cite{AMS09,EAS98}, 
\begin{equation}\label{def:bm}
\min_{\BS_i\in \St(d,p)}~-\sum_{i,j} \lag \BC_{ij}, \BS_i\BS_j^{\top}\rag \tag{BM}
\end{equation}
This family of optimization programs~\eqref{def:bm} with different $p$ is an ``interpolation" between~\eqref{def:od} and~\eqref{def:sdr}: 
if $p = d$,~\eqref{def:bm} equals~\eqref{def:od} and if $p = nd$,~\eqref{def:bm} is exactly the~\eqref{def:sdr}. 
Despite the nonconvexity of~\eqref{def:bm}, it has been observed that even if the rank is quite small, the algorithm still does not often get stuck at spurious local minima even with random initialization. Therefore, we are interested in the following question:
\begin{equation*}
\text{\em  Are there no other local minima in~\eqref{def:bm} besides the global minima? }
\end{equation*}

It has been shown that no spurious local minima exist for a class of the SDPs, i.e., the optimization landscape is {\em benign}, if the rank of the factorization $p$ is sufficiently large, namely greater than the square root of the number of the constraints~\cite{BVB20,WW20}. 
Many works have shown that the landscape of the nonconvex objective function is benign under nearly identical conditions that enable the success of the SDR. Examples include phase retrieval~\cite{SQW18}, phase synchronization and community detection~\cite{L23c}, orthogonal group synchronization~\cite{L20a}, matrix completion~\cite{GLM16} and dictionary learning~\cite{SQW16}.  Therefore, we aim to provide a bound on the rank of the factorization that theoretically guarantees a benign optimization landscape for the GOPP under arbitrary noise.


\section{Main results}\label{s:main}
We denote vectors and matrices by boldface letters $\bx$ and $\BX$ respectively. For a given matrix $\BX$, $\BX^{\top}$ is the transpose of $\BX$ and $\BX\succeq 0$ means $\BX$ is symmetric and positive semidefinite. Let $\I_n$ be the identity matrix of size $n\times n$, and $\BJ_n$ and $\bone_n$ be the $n\times n$ matrix and $n\times 1$ vector with all entries equal to 1. For two matrices  $\BX$ and $\BY$ of the same size, their inner product is $\lag \BX,\BY\rag= \Tr(\BX^{\top}\BY) = \sum_{i,j}X_{ij}Y_{ij}.$ Let $\|\BX\|$ be the operator norm and $\|\BX\|_F$ be the Frobenius norm. We denote the $i$th largest and the smallest singular value of $\BX$ by $\sigma_{i}(\BX)$ and $\sigma_{\min}(\BX)$ respectively.
For a non-negative function $f(x)$, we write $f(x)\lesssim g(x)$ if there exists a positive constant $C_0$ such that $f(x)\leq C_0g(x)$ for all $x.$

\subsection{Main theorems}\label{ss:main}
{Before we present the main theorems, we provide some additional assumption on the data matrix $\BA$. Given a matrix $\BA$, we let  $\sigma_{\max}(\BA) := \sqrt{\lambda_{\max}(\BA\BA^{\top})}$ and $\sigma_{\min}(\BA) := \sqrt{\lambda_{\min}(\BA\BA^{\top})}$, and $\lambda_{\max}(\cdot)$ and $\lambda_{\min}(\cdot)$ are the largest and smallest eigenvalues respectively. In particular, we assume $\BA$ is full row-rank, i.e., $\sigma_{\min}(\BA) > 0$. Let $\kappa$ be the condition number of $\BA$, i.e., $\kappa = \sigma_{\max}(\BA)/\sigma_{\min}(\BA)$.} 
Define the SNR (signal-to-noise ratio) as 
\begin{equation}\label{def:snr}
\SNR := \frac{\sigma_{\min}(\BA)}{\max_{1\leq i\leq n} \|\BDelta_i\|}.
\end{equation}
Now we present our first main result on the tightness of the SDR.

\begin{theorem}[{\bf Tightness of the SDR}]\label{thm:main}
Suppose 
\[
\SNR \geq 30\kappa^3\sqrt{d},
\]  
the~\eqref{def:sdr} is tight and recovers the unique global minimizer  to~\eqref{def:od}, i.e., the global minimizer to the SDR is unique and exactly rank-$d$.
\end{theorem}

Theorem~\ref{thm:main} does not assume any statistical models on $\BDelta$. The sufficient condition for the tightness is completely~\emph{algebraic} and works for any type of additive adversaries. In other words, solving~\eqref{def:sdr} yields the least squares estimator to the generalized Procrustes problem in the high SNR regime. 

{Before we proceed to present the convergence of the GPM, we introduce one more notation.
Note that the value of~\eqref{def:od} remains the same if we multiply $\BX\in\Od(d)^{\otimes n}$ by another orthogonal matrix from the right side. We resolve this ambiguity by using the following distance for two elements $\BX$ and $\BY\in\RR^{nd\times d}$ by
\begin{equation}\label{def:df}
d_F(\BX,\BY) := \min_{\BQ\in\Od(d)} \|\BX - \BY\BQ\|_F = \min_{\BQ\in\Od(d)} \sqrt{ \sum_{i=1}^n\|\BX_i - \BY_i\BQ\|_F^2}
\end{equation}
where the minimizer $\BQ =\PP(\BY^{\top}\BX)$ is given in~\eqref{def:proj}. One can verify { $d_F(\BX,\BY)$  in~\eqref{def:df}} is indeed a distance and satisfies triangle inequality.}


\begin{theorem}[{\bf Convergence of the GPM}]\label{thm:gpm}
Suppose 
\[
\SNR \geq 35\kappa^4\sqrt{d},
\] 
the generalized power method with spectral initialization recovers the globally optimal solution to the SDR with linear convergence
\[
d_F(\BO^t, \widehat{\BO}) \leq \rho^t d_F(\BO^0, \widehat{\BO}),~~\rho < 1,
\]
where $\widehat{\BO}$ and $\widehat{\BO}\widehat{\BO}^{\top}$ are the unique globally optimal solution to~\eqref{def:od} and~\eqref{def:sdr} respectively. Moreover, {it holds that}
\[
\min_{\BQ\in\Od(d)} \|\widehat{\BO} - \BO\BQ\|_F \leq 7\kappa^2 \SNR^{-1}\sqrt{nd}.
\]
\end{theorem}
Theorem~\ref{thm:gpm} gives an efficient algorithm to produce the globally optimal solution to~\eqref{def:od} and~\eqref{def:sdr}. Moreover, the sufficient conditions in Theorem~\ref{thm:main} and~\ref{thm:gpm} are almost the same. {Both theorems are general and can be applied to different settings, e.g., $\BDelta_i$ is a random subgaussian matrix in which case $\|\BDelta_i\|$ is straightforward to upper bound, see~\cite{V18}. We will provide two concrete examples in the next sections.}
Finally, we will discuss the optimization landscape of~\eqref{def:bm}, whose proof is characterized via Theorem 2.8 in~\cite{L23c}.

\begin{theorem}[{\bf Optimization landscape of the Burer-Monteiro factorization}]\label{thm:bm}
Suppose that
\[
\SNR\geq  \frac{7\kappa^2 \sqrt{d}(p+2\kappa^2 d + d-2)}{p-2\kappa^2 d+d-2}~~ \text{ and }~~\SNR\geq 35\kappa^4\sqrt{d}, 
\]
then any second order critical point $\BS\in\St(d,p)^{\otimes n}$ is global and $\BS\BS^{\top}$ is equal to the unique global minimizer to~\eqref{def:sdr}. In particular, $p\geq 2\kappa^2 d + 2$ along with $\SNR\geq 35\kappa^4\sqrt{d}$ suffices to ensure a benign optimization landscape.
\end{theorem}
Theorem~\ref{thm:bm} justifies why gradient-based approach do not suffer from getting stuck at any local minimizers in~\eqref{def:bm} as every local minimizer in~\eqref{def:bm} is a global minimizer to~\eqref{def:sdr} provided that $p$ is slightly larger than $d$ in the high SNR regime.

\subsection{Applications in two statistical models}
\paragraph{Example: Gaussian noise model.} Consider
\begin{equation}\label{def:modelgauss}
\BA_i = \BO_i \BA+ \sigma \BW_i, \quad 1\leq i\leq n,
\end{equation}
where $\BW_i\in\RR^{d\times m}$ is a Gaussian random matrix.
In this case, it suffices to have
\[
\SNR : = \frac{\sigma_{\min}(\BA)}{\max \|\BDelta_i\|} \geq \frac{\sigma_{\min}(\BA)}{\sigma (\sqrt{d} + \sqrt{m} + \sqrt{2\gamma\log n})} \geq 35\kappa^4\sqrt{d}
\]
where $\max\|\BW_i\| \leq \sqrt{d} + \sqrt{m} + \sqrt{2\gamma\log n}$ holds with probability at least $1-2n^{-\gamma+1}$, following from~\cite[Theorem 2.26]{W19} and~\cite[Theorem 7.3.1]{V18}.

\begin{corollary}\label{cor:gauss}
Consider the model~\eqref{def:modelgauss} and assume
\begin{equation}\label{eq:bks}
\sigma \leq  \frac{1}{35\kappa^4\sqrt{d}}\cdot\frac{\sigma_{\min}(\BA)}{ \sqrt{d} + \sqrt{m} + \sqrt{2\gamma\log n} }
\end{equation}
the GPM with spectral initialization converges linearly to the global minimizer to~\eqref{def:sdr} with probability at least $1-2n^{-\gamma+1}.$
\end{corollary}

Under the same setting,~\cite[Theorem 3.2]{L23d} proves that 
\begin{equation}\label{eq:l23d}
\sigma\lesssim \frac{1}{\kappa^4\sqrt{d}}\cdot \frac{\sigma_{\min}(\BA)}{\sqrt{d} + \sqrt{m/n} + \sqrt{2\gamma \log n}}
\end{equation}
is needed for the global convergence of the GPM and the tightness of SDR.
Our bound~\eqref{eq:bks} is sub-optimal particularly if $m$ is very large, in which increasing $n$ does not help improve the bound in~\eqref{eq:bks} unlike~\eqref{eq:l23d}.

\paragraph{Example: uniform corruption model} Consider  $\BA = [\ba_1,\cdots,\ba_m]$
consists of $m$ points on $\mathbb{S}^{d-1}.$ Each point cloud $\BA_i$ is constructed as follows:
\begin{equation}\label{def:modelbinary}
\ba_{ij} = \varphi_{ij} \ba_j + (1-\varphi_{ij})\bz_{ij},~~~~\varphi_{ij}\overset{i.i.d.}{\sim}\text{Ber}(\theta),~~\bz_{ij}\sim\text{Unif}(\mathbb{S}^{d-1})
\end{equation}
where $\ba_{ij}$ denotes the $j$-th point in the $i$-th point cloud.
The matrix form of the $i$-th point cloud equals
\[
\BA_i = \theta\BA + \BA(\diag(\bvarphi_i) - \theta\I_m) + \BZ_i (\I_m - \diag(\bvarphi_i))
\]
where $\bvarphi_i = [\varphi_{i1},\cdots,\varphi_{im}]^{\top}.$
The noisy part of $\BA_i$ can be written into:{
\begin{align*}
\BDelta_i & = \BA(\diag(\bvarphi_i) - \theta\I_m) + \BZ_i (\I_m - \diag(\bvarphi_i)) =  \sum_{j=1}^m\BDelta_{ij}, \\
\BDelta_{ij} & = \left((\varphi_{ij} -\theta) \ba_j + (1-\varphi_{ij})\bz_{ij}\right)\be_j^{\top}.
\end{align*}}
To apply Theorem~\ref{thm:main} and~\ref{thm:gpm}, it suffices to estimate $\|\BDelta_i\|$ and the singular values of $\BA$, which is quite standard via the matrix Bernstein inequality~\cite[Theorem 1.4]{T12}. We leave the proof to Section~\ref{ss:corbinary}.

\begin{corollary}\label{cor:binary}
Consider the model~\eqref{def:modelbinary} and assume
\begin{equation}
1\geq \theta \geq \max\left\{  \frac{400\kappa^4 d\log (n(m+d))}{\sqrt{m}}, 1 - \frac{1}{(140\kappa^4\sqrt{12d\log (n(m+d))})^2} \right\}
\end{equation}
and{
\[
m \geq 400^2\kappa^8 d^2 \log^2(n(m+d)),
\]
}the GPM with spectral initialization converges linearly to the global minimizer to~\eqref{def:sdr} with probability at least $1-2(n(m+d))^{-1}.$
\end{corollary} 
We proceed to perform some experiments to see if Corollary~\ref{cor:binary} is tight empirically, and then we will discuss more on the optimality and sub-optimality of our result. 

For the experiments, we will run the generalized power method with~\emph{random} initialization $\BO^0$. Once the iterates $\{\BO^t\}$ stabilize, we verify its global optimality via convex analysis. More precisely, 
the algorithm stops if
\[
\| \BO^{t+1} (\BO^{t+1})^{\top} -  \BO^t (\BO^t)^{\top}\|_F \leq 10^{-6}
\]
or reaches the maximum number of iterations. We consider $\BO^t$ as an approximate global minimizer if
\[
\| (\BLambda^t - \BC)\BO^t \| < 10^{-6}, \qquad \lambda_{d+1}(\BLambda^t - \BC) > 0
\]
where $\BLambda^t$ is a symmetric block-diagonal matrix whose $i$th block equals $\BLambda_{ii}^t =( [\BC\BO^{t}]_i[\BC\BO^{t}]_i^{\top})^{1/2}$ and $\lambda_{d+1}(\BLambda^t-\BC)$ denotes the $(d+1)$th smallest eigenvalue of $\BLambda^t-\BC$.

For the additive Gaussian model, we do not repeat the experiments here: the simulations in~\cite{L23d} indicate that if $\sigma_{\min}(\BO\BA )> 1.89 \|\BW\|$ where $\|\BW\| \leq \sqrt{nd} + \sqrt{m} + 2\sqrt{n\log n}$, i.e., then the GPM converges to the global minimizer of the SDR with high probability. 
For the uniform corruption model, we set different $(d,n,m)$ in the experiment, and $\theta$ varies from 0 to 1. For each set of $(d,n,m,\theta)$, we run 20 experiments and calculate the proportion of the global convergence of the~\eqref{def:gpm}, as shown in Figure~\ref{fig:1}. 
We can see that the phase transition boundary between the black region (failure) and the white one (success) is quite sharp. To see how the signal-to-noise affects the convergence, we plot $\|\BO\BA\|/\|\BDelta\|$ v.s. the frequency of success in the right column. The empirical experiments show that $\|\BO\BA\|/\|\BDelta\| > 2$, the tightness holds. 
\begin{figure}[h!]
\centering
\begin{minipage}{0.48\textwidth}
\centering
\includegraphics[width=65mm]{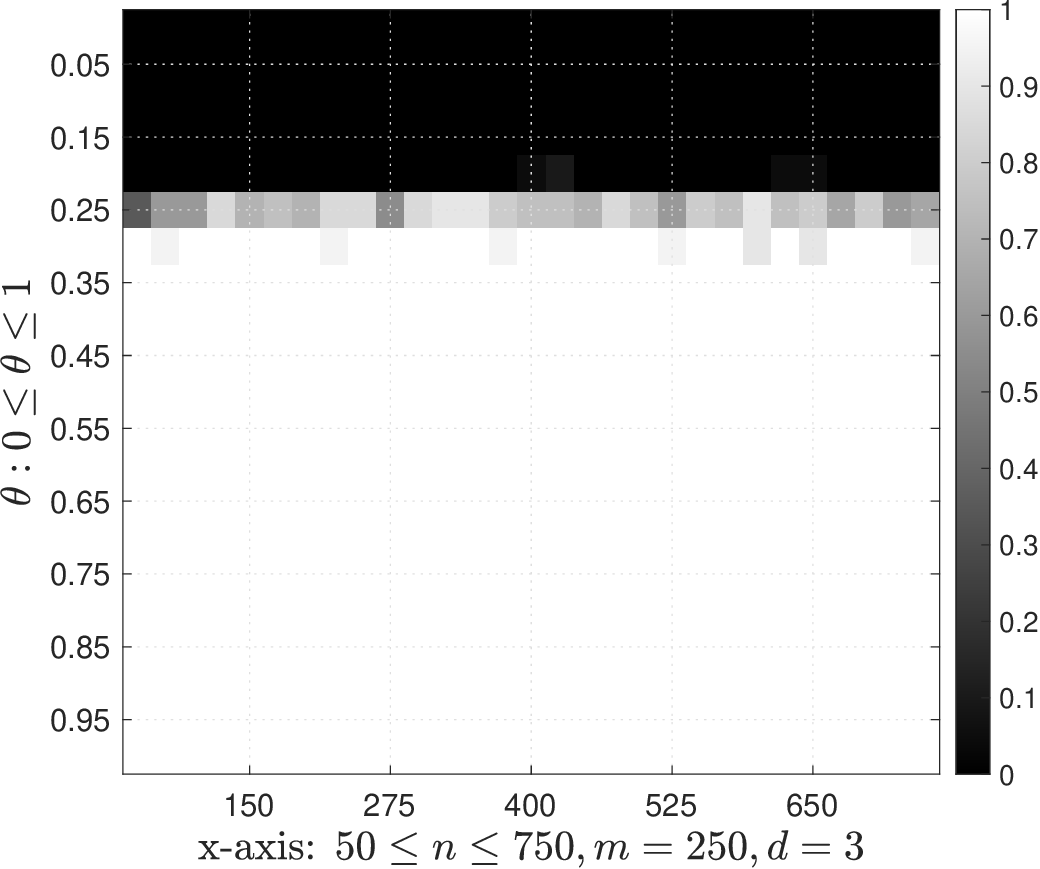}
\end{minipage}
\hfill
\begin{minipage}{0.48\textwidth}
\includegraphics[width=60mm]{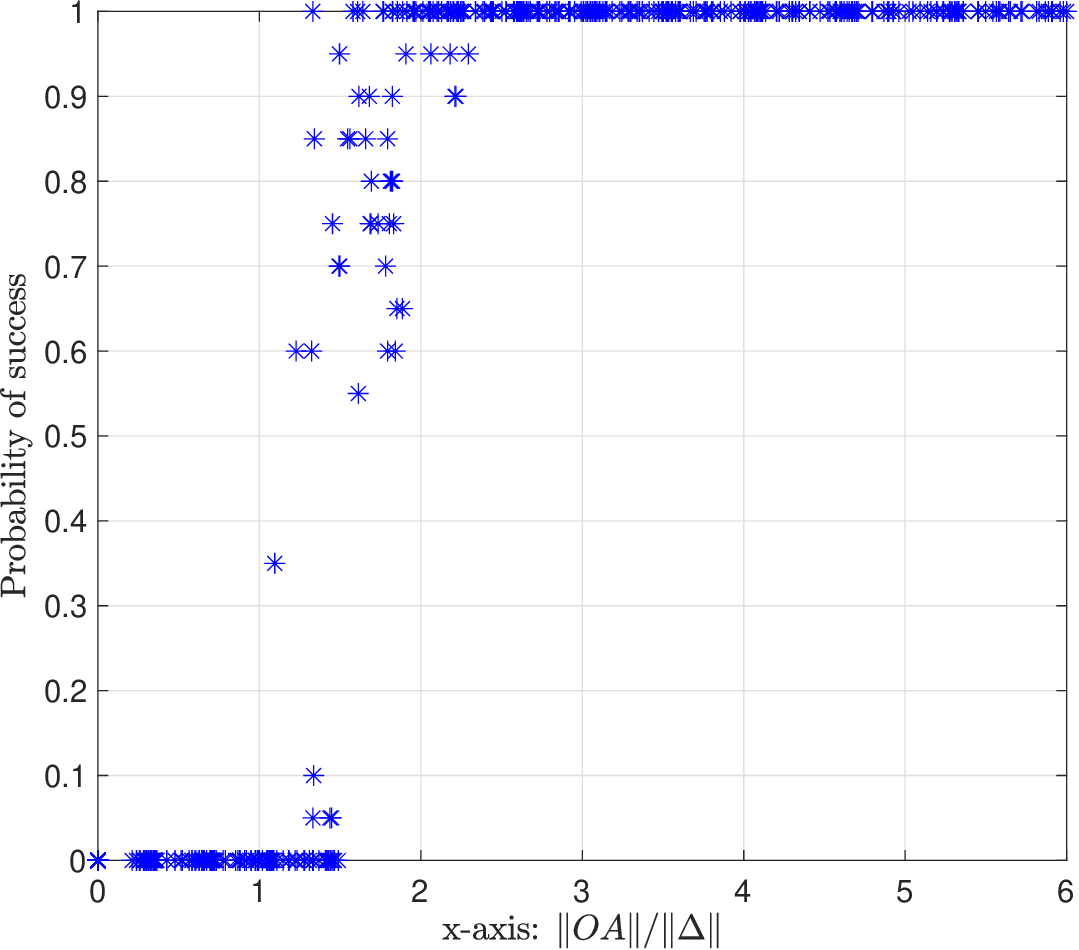}
\end{minipage}
\vfill
\vskip0.2cm
\begin{minipage}{0.48\textwidth}
\centering
\includegraphics[width=65mm]{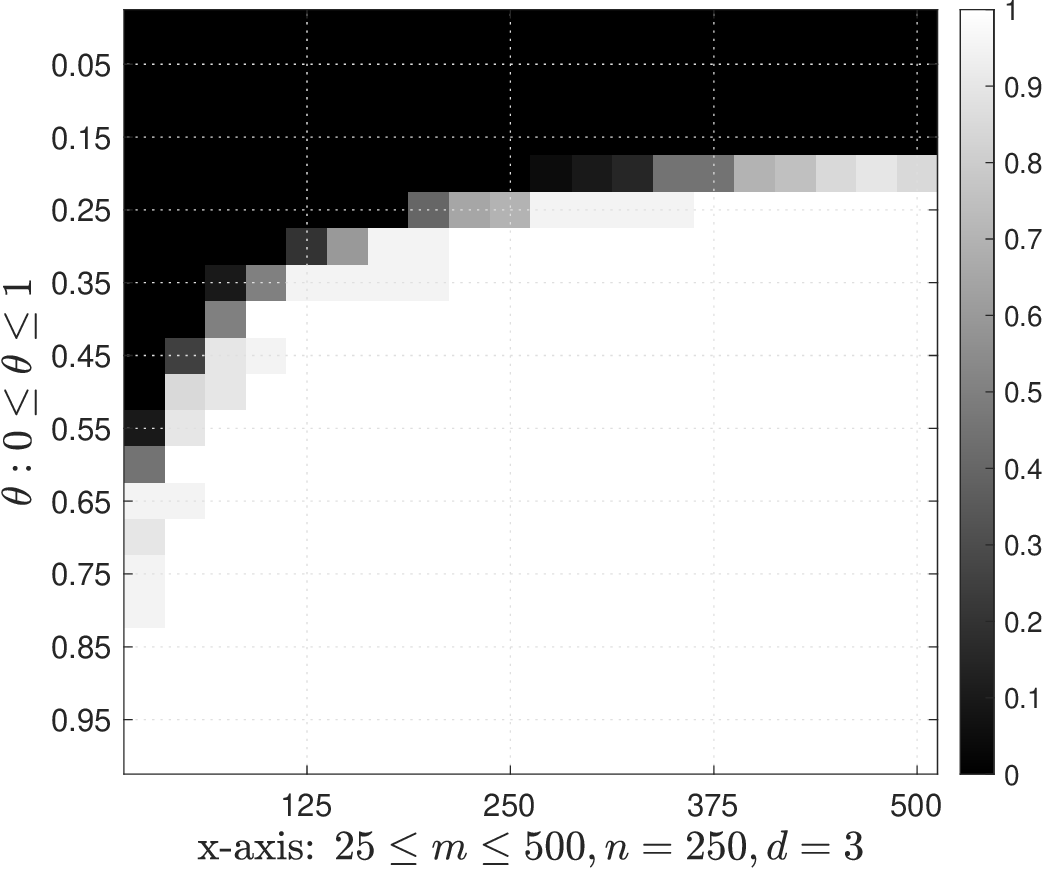}
\end{minipage}
\hfill
\begin{minipage}{0.48\textwidth}
\includegraphics[width=60mm]{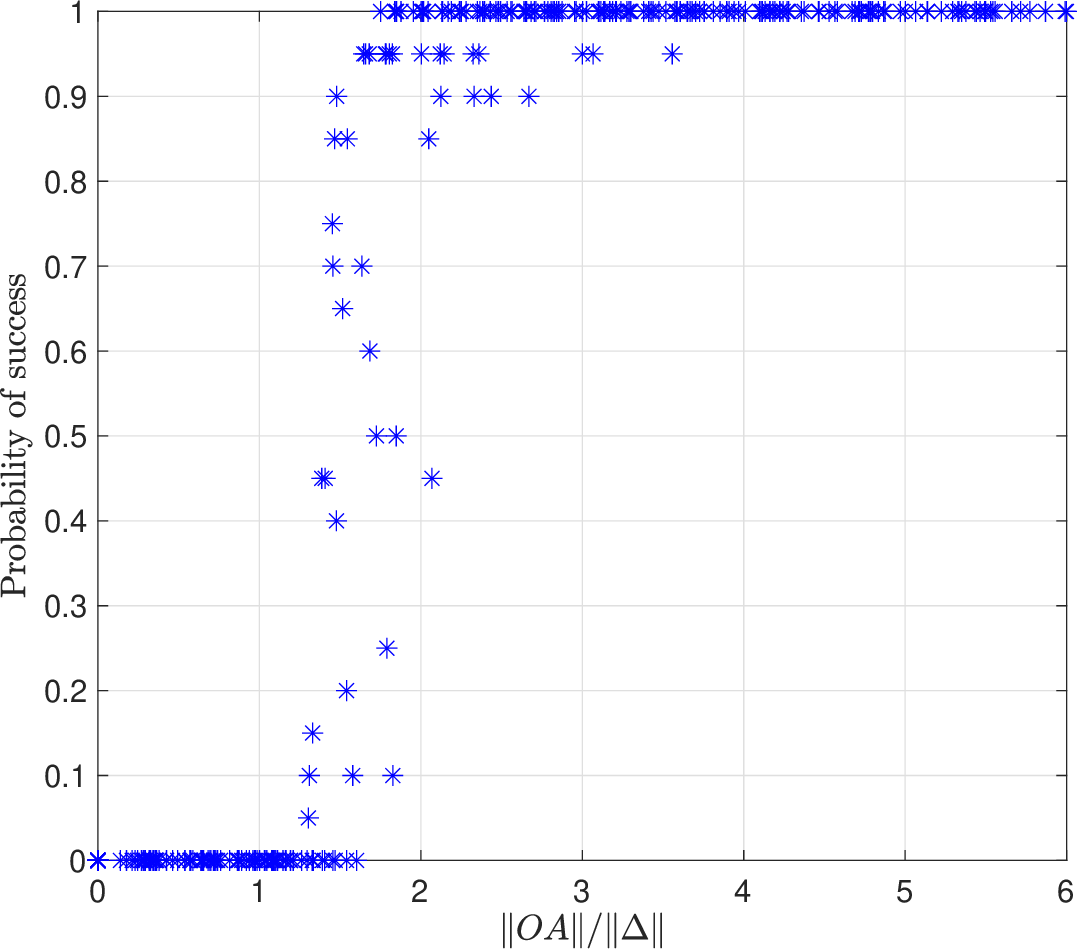}
\end{minipage}
\vfill
\vskip0.2cm
\begin{minipage}{0.48\textwidth}
\centering
\includegraphics[width=65mm]{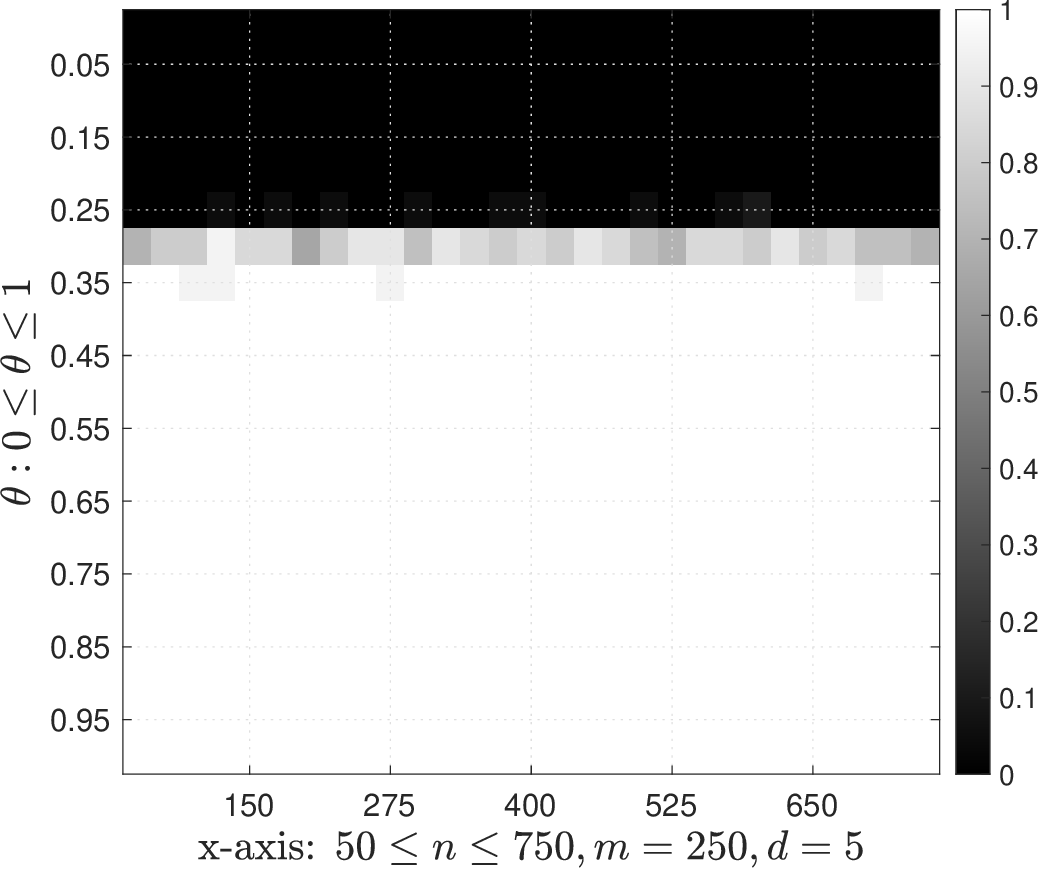}
\end{minipage}
\hfill
\begin{minipage}{0.48\textwidth}
\includegraphics[width=60mm]{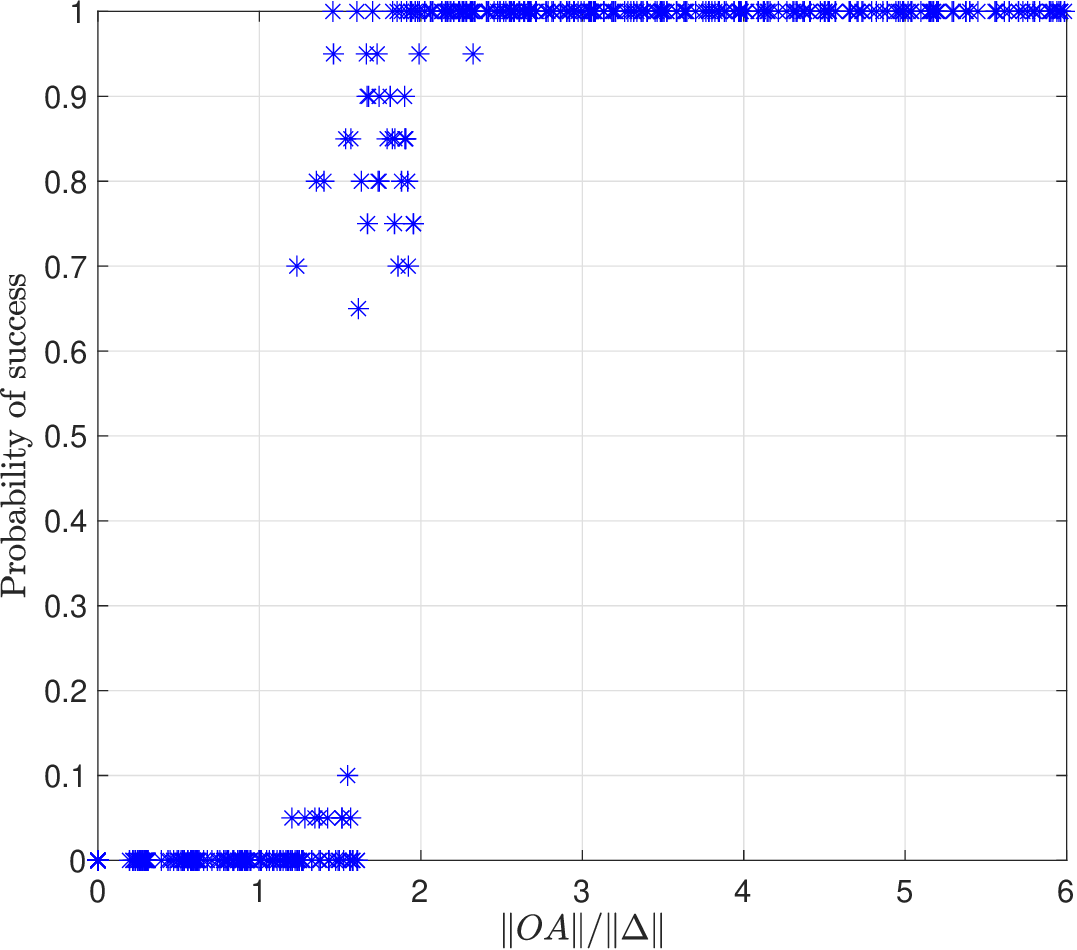}
\end{minipage}
\vfill
\vskip0.2cm
\begin{minipage}{0.48\textwidth}
\centering
\includegraphics[width=65mm]{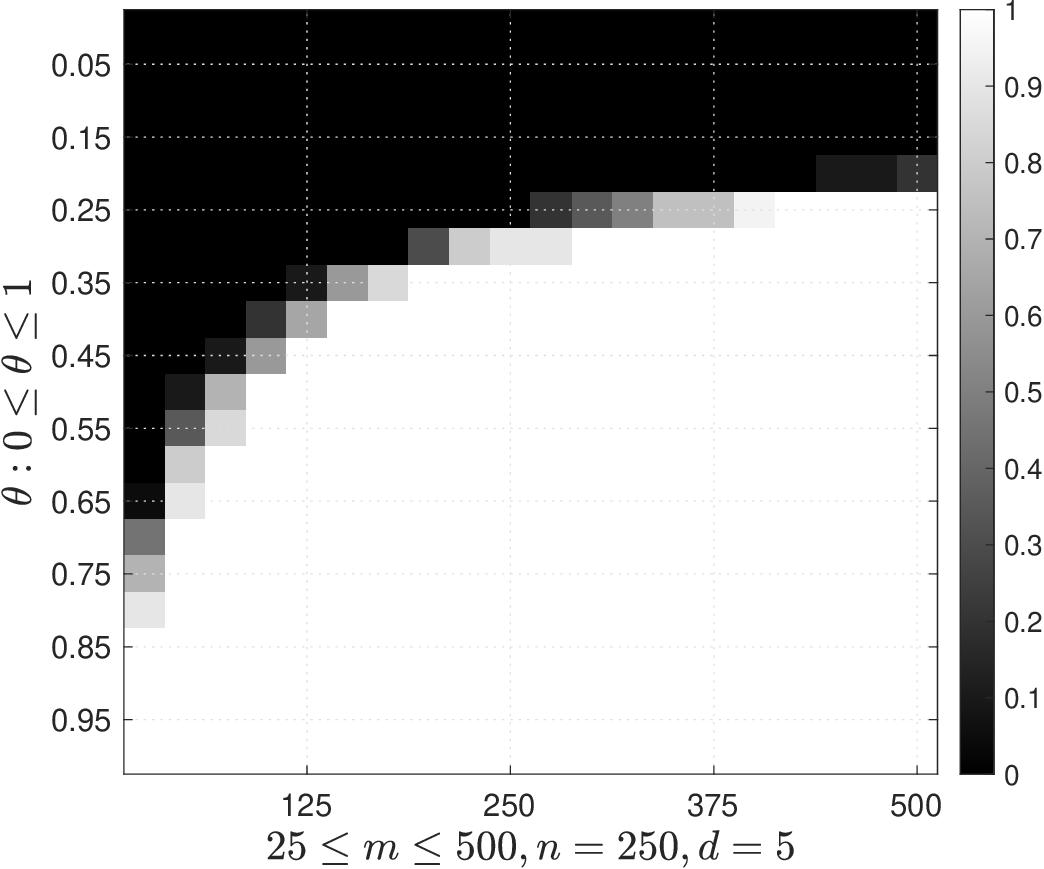}
\end{minipage}
\hfill
\begin{minipage}{0.48\textwidth}
\includegraphics[width=60mm]{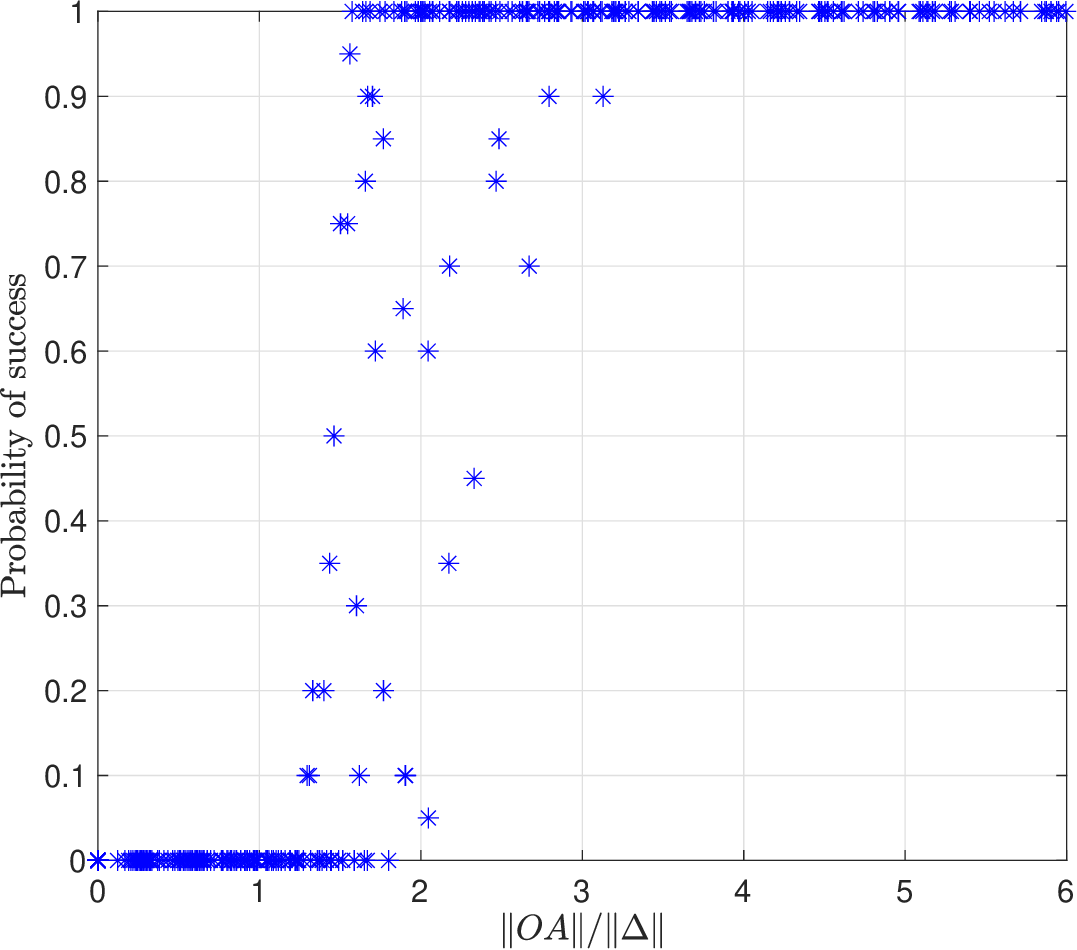}
\end{minipage}
\caption{Phase transition for the tightness/global convergence of the GPM. Left column: black region: the tightness of the~\eqref{def:gpm} fails; white region: the tightness holds; Right column: the frequency of success v.s. $\|\BO\BA\|/\|\BDelta\|$ }
\label{fig:1}
\end{figure}

\vskip0.2cm
Here we provide a brief discussion on how optimal our result is.
Given $\BD = \BO\BA + \BDelta$, we define another notion of the {signal-to-noise} ratio as:{
\begin{equation}\label{def:snrs}
\SNR_S : = \frac{\|\BO\BA\|}{\|\BDelta\|} = \frac{\sqrt{n}\|\BA\|}{\|\BDelta\|}
\end{equation}
where $\|\BO\BA\| = \sqrt{n}\|\BA\|.$}
This notion of the SNR arises naturally in bounding the estimator of $\BO$ via spectral method, i.e., extracting the top $d$ left singular values $\BU$ of $\BD$ with $\BU^{\top}\BU = n\I_d$, and then it holds from Davis-Kahan theorem that 
\[
\frac{1}{\sqrt{n}}\min_{\BQ\in\Od(d)}\|\BU - \BO\BQ\| \leq \frac{2\|\BDelta \|}{\sqrt{n}\sigma_{\min}(\BA) - \|\BDelta\|} = \frac{2}{\kappa^{-1}\SNR_S - 1}.
\]
This implies that the spectral method provides a non-trivial error bound if $\SNR_S$ is at least greater than a constant. In our numerical experiments indicate that if $\SNR_S \geq 2$ for both two models, then the global convergence of the GPM holds. Now we proceed to compare $\SNR$ and $\SNR_S$ whose difference is mainly on the measure of noise. 
Let{
\[
\xi : = \frac{\sqrt{n}\max\|\BDelta_i\|}{\|\BDelta\|}\geq 1~~ \text{ and then }~~ \SNR_S = \frac{\sqrt{n}\|\BA\|}{\sqrt{n}\xi^{-1}\max\|\BDelta_i\|} = \kappa\xi \SNR.
\]
The reason why $\xi$ is defined in this form is: if $\xi$ is of order 1, then $\SNR_S$ and $\SNR$ are of the same order. In particular, if $\{\BDelta_i\}_{i=1}^n$ are i.i.d. Gaussian random matrices and $m/d$ is not too large, then $\xi$ is close to 1. The factor $\xi$ determines how much $\SNR$ is different from $\SNR_S.$
}

Based on  $\SNR_S\geq 2$ from our numerical experiments, we have a bound in terms of the $\SNR$:
\[
\SNR = \frac{\SNR_S}{\kappa\xi} \geq \frac{2}{\kappa\xi}
\]
which differs from our bound in Theorem~\ref{thm:main} and~\ref{thm:gpm} by a factor that involves $\xi, \kappa$ and $d$. The main difference comes from $\xi$ as it may vary significantly in specific cases. In particular, if $\xi$ is large, then our bound is far away from the phase transition boundary in the simulations. {Now we roughly approximate a lower bound of $\xi$ in the two applications. Suppose $\xi_-$ is a lower bound, then as long as $\SNR \geq \SNR_S/(\kappa \xi_-) \geq 2/(\kappa \xi_-)$, the tightness of SDR holds which follows from
\[
\SNR_S = \kappa \xi \SNR \geq \frac{2\kappa \xi}{\kappa \xi_-} \geq 2.
\]
In both models, if $m\gg d$ and $n$ are large, then $\xi_-$ is large and our bound is sub-optimal.}

\paragraph{Gaussian noise model}{
For $\BDelta=\sigma\BW$, then we have $\|\BW\| \leq \sqrt{m} + \sqrt{nd} + 2\sqrt{n\log n}$, and $\max\|\BW_i\| \geq  \sqrt{m} - \sqrt{d}- 2\sqrt{\log n}.$
Then
\[
\xi \geq \frac{ \sqrt{m} - \sqrt{d} - 2\sqrt{\log n}}{\sqrt{m/n}+\sqrt{d}  + 2\sqrt{\log n}} \gtrsim \min \left\{ \sqrt{n}, \sqrt{\frac{m}{d}} \right\}
\]
for some constant $C>0$ where $m\geq d+1$. 
 In other words, if $m/d$ and $n$ are large, then our bound is not tight.}

\paragraph{Uniform corruption model} {
For the uniform corruption model~\eqref{def:modelbinary}, we note that
\[
\E \BDelta^{\top}\BDelta = n(1-\theta^2)\I_m,~~~\E \BDelta_i\BDelta_i^{\top} = (1-\theta)\left(\theta\BA\BA^{\top} + \frac{m}{d}\I_d\right).
\]
For $m$ and $n$ sufficiently large, $\BDelta^{\top}\BDelta$ and $\BDelta_i\BDelta_i^{\top}$ will not deviate from their respective expectations much. As a result, we can argue that $\|\BDelta_i\|$ has a lower bound of order at least  $(1-\theta)\sqrt{m/d}$ and $\|\BDelta\|$ is of order at most $\sqrt{(1-\theta^2)n}.$
Therefore, we have
\[
\xi = \frac{\sqrt{n}\max\|\BDelta_i\|}{\|\BDelta\|} \gtrsim \sqrt{\frac{m}{d}}.
\]
This implies that if $m/d$ and $n$ are large, then $\xi\gg 1$ and our theorem is sub-optimal.}

\section{Proofs}\label{s:proof}

Before we proceed to the proof, we first introduce a few notations and two useful theorems that will be used later. 
The distance {$d_F(\cdot,\cdot)$ in~\eqref{def:df}} can be computed exactly by using
\begin{equation*}
d_F^2(\BX,\BY) = \|\BX\|_F^2 + \|\BY\|_F^2- 2\|\BX^{\top}\BY\|_*
\end{equation*}
where $\|\cdot\|_*$ denotes the matrix nuclear norm, i.e., the sum of singular values.
In particular, if $\BX$ and $\BY$ are both in $\Od(d)^{\otimes n}$, then
\begin{equation}\label{def:df2}
d_F^2(\BX,\BY) = 2nd - 2\|\BX^{\top}\BY\|_*
\end{equation}

To measure how close a given element $\BS\in\Od(d)^{\otimes n}$ is to the ground truth state, i.e., $\BO\BQ\in\Od(d)^{\otimes n}$ for some $\BQ\in\Od(d,p)$, we define an $\eps$-neighborhood around $\BO$ in $\Od(d)^{\otimes n}\subset\RR^{nd\times d}$:{
\begin{equation}\label{def:neps}
\mathcal{N}_{\eps} : = \left\{ \BS\in\Od(d)^{\otimes n}: d_F(\BS,\BO)\leq \eps\sqrt{nd}\right\}.
\end{equation}}
For any $\BS\in\mathcal{N}_{\eps}$, we have
\begin{equation}\label{eq:sigmaZS}
\left(1-\frac{\eps^2d}{2}\right)n \leq \sigma_{\min}(\BS^{\top}\BO) \leq \sigma_{\max}(\BS^{\top}\BO)\leq n
\end{equation}
which follows from $2nd - 2\|\BS^{\top}\BO\|_* \leq \eps^2 nd$ and $\sigma_{\max}(\BS^{\top}\BO)\leq n$: { note that
\[
\frac{\eps^2 nd}{2} \geq nd - \sum_{i=1}^d \sigma_i(\BS^{\top}\BO) = \sum_{i=1}^d (n - \sigma_i(\BS^{\top}\BO)) \geq n - \sigma_{\min}(\BS^{\top}\BO)
\] 
since $0\leq\sigma_{i}(\BS^{\top}\BO)\leq n.$ Then we have a lower bound on the smallest singular value of $\BS^{\top}\BO$ for any $\BS\in\mathcal{N}_{\eps}.$
}

One important supporting theorem is the optimality condition of~\eqref{def:sdr}. This result provides a simple sufficient condition to certify whether a candidate solution is a global minimizer to~\eqref{def:od} and~\eqref{def:sdr}.
\begin{theorem}[\bf Characterization of global optimality]\label{thm:cvx}
The matrix {$\BZ = \BS\BS^{\top}$}with $\BS\in \Od(d)^{\otimes n}$ is a global minimizer to~\eqref{def:sdr} if there exists a block-diagonal matrix \[
\BLambda = \blkdiag(\BLambda_{11},\cdots,\BLambda_{nn})\in\RR^{nd\times nd}
\] 
with $\BLambda_{ii}\in\RR^{d\times d}$ such that
\begin{equation}\label{cond:opt}
\BC\BS = \BLambda \BS, \quad \BLambda - \BC\succeq 0.
\end{equation}
Moreover, if $\BLambda - \BC$ is of rank $(n-1)d$, then $\BX$ is the unique global minimizer, and $\BS$ is rank-$d$ and equals the global minimizer to~\eqref{def:od}.
\end{theorem}
This optimality condition can be found in several places including~\cite[Theorem 3.4]{B15},~\cite[Proposition 5.1]{L20a} and~\cite[Theorem 7]{RCBL19}, whose proof follows from the standard routine of duality theory in convex optimization. 
The other supporting result is the very well-known contraction principle, which will be used to analyze the convergence of the generalized power method.

\begin{theorem}\label{thm:banach}{\bf Banach contraction principle,~\cite[Theorem 1.1]{GD03}}
Let $({\cal X},d)$ be a complete metric space and ${\cal T}: {\cal X}\rightarrow {\cal X}$ be a contraction map, i.e.,
\[
d({\cal T}(x),{\cal T}(y))\leq \rho d(x,y), \quad\forall x,y\in {\cal X}
\] 
for an absolute positive constant $\rho<1$.
Then there is a unique fixed point $u$ and ${\cal T}^t(x)$ converges to $u$ as $t\rightarrow \infty$ for each $x\in {\cal X}$ with
\[
d({\cal T}^t(x), u) \leq \rho^n d(x, u).
\]
\end{theorem}

\subsection{An overview of the proof}\label{ss:gpm}
Define the operator ${\cal T}$ on $\RR^{nd\times d}$:
\begin{equation}\label{def:T}
{\cal T}(\BX) : = {\cal P}_n (\BC \BX),
\end{equation}
for any $\BX\in\Od(d)^{\otimes n}$
and Algorithm~\ref{algo1} becomes $\BO^{t+1} = {\cal T}(\BO^t)$.
Under certain conditions on $\eps$ and $\SNR$, we will first prove ${\cal T}$ is a contraction mapping on $\mathcal{N}_{\eps}$; moreover, for any $\BX\in\mathcal{N}_{\eps}$, ${\cal T}(\BX)$ still stays in $\mathcal{N}_{\eps}$ and ${\cal T}^t(\BX)$ converges to a unique fixed point and the limiting point corresponds to the global minimizer to the SDR. 

For the tightness of SDR and the convergence of generalized power method, we consider two sets of conditions on $\eps$ and $\SNR$.
\begin{itemize}
\item Condition ($a$) for the tightness of SDR:
\begin{equation}\label{def:casea}
\eps \sqrt{d} = 0.19\kappa^{-2},~~\SNR\geq 30\kappa^3\sqrt{d}.
\end{equation}
\item Condition ($b$) for the convergence of GPM with spectral initialization:
\begin{equation}\label{def:caseb}
\eps \sqrt{d} = 7\kappa^2\sqrt{d} \SNR^{-1},~~\SNR\geq 35\kappa^4\sqrt{d}.
\end{equation}
\end{itemize}

\begin{proposition}\label{prop:key}
${\cal T}$ is a contraction mapping from $\mathcal{N}_{\eps}$ to $\mathcal{N}_{\eps}$ under both~\eqref{def:casea} and~\eqref{def:caseb}. 
\end{proposition}
The proof of Proposition~\ref{prop:key} is a direct consequence of Proposition~\ref{prop:con} and~\ref{prop:basin}.

\begin{proposition}[\bf Lipschitz continuity of ${\cal T}$ on $\mathcal{N}_{\eps}$]\label{prop:con}
For $\BX$ and $\BY\in\mathcal{N}_{\eps}$, it holds that
\begin{align*}
d_F({\cal T}(\BX), {\cal T}(\BY)) 
& \leq  \rho d_F(\BX,\BY)
\end{align*}
where
\begin{equation}\label{cond:con}
\rho(\eps,\SNR) :=  \frac{2\kappa^2\left(2\eps\sqrt{d} + (2+ \SNR^{-1}) \SNR^{-1} \right) }{ 1 - \left(\eps^2d/2 +  \kappa(2 + \SNR^{-1})\SNR^{-1}\right) }.
\end{equation}

\end{proposition}

\begin{proposition}[\bf ${\cal T}$ maps $\mathcal{N}_{\eps}$ to itself]\label{prop:basin}
Suppose
\begin{equation}\label{cond:basin}
\eta(\eps,\SNR): = \frac{4\kappa^2\left(\eps^2 d/2  + \kappa^{-1}\sqrt{d}(2+\SNR^{-1})\SNR^{-1}\right) }{  2 - \left( \eps^2d/2 + \kappa (2 + \SNR^{-1})\SNR^{-1}\right) } \leq \eps\sqrt{d} 
\end{equation}
then we have ${\cal T}(\BX)\in\mathcal{N}_{\eps}$ for any $\BX\in\mathcal{N}_{\eps}$.
\end{proposition}

\begin{proof}[\bf Proof of Proposition~\ref{prop:key}]
Proposition~\ref{prop:con} and~\ref{prop:basin} imply that ${\cal T}$ is a contraction mapping on $\mathcal{N}_{\eps}$ if $\rho(\eps,\SNR) < 1$ and $\eta(\eps,\SNR) \leq \eps\sqrt{nd}$. We will now verify them under~\eqref{def:casea} and~\eqref{def:caseb}.

For Condition $(a)$ in~\eqref{def:casea}, we have $\eps\sqrt{d} = 0.19\kappa^{-2}$ and $\SNR\geq 30\kappa^3\sqrt{d} \geq 30$. It holds that{
\begin{align*}
& \rho \leq \frac{2\kappa^2\left(0.38\kappa^{-2} + (2+ 30^{-1}\kappa^{-3}) 30^{-1}\kappa^{-3} \right) }{ 1 - 0.19^2\kappa^{-4}/2 -  \kappa(2 + 30^{-1}\kappa^{-3})30^{-1}\kappa^{-3} } 
{\leq \frac{2\left(0.38 + (2+ 30^{-1}) 30^{-1} \right) }{ 1 - 0.19^2/2 -  (2 + 30^{-1})30^{-1}}  < 1},\\
& \eta \leq \frac{4\kappa^{-2}\left(0.19^2/2  + 30^{-1}(2+30^{-1}\kappa^{-3}) \right) }{  2 - \left( 0.19^2\kappa^{-4}/2 + 30^{-1}\kappa^{-2} (2 +30^{-1}\kappa^{-3})\right) } 
{\leq \frac{4\kappa^{-2}\left(0.19^2/2  + 30^{-1}(2+30^{-1}) \right) }{  2 - \left( 0.19^2/2 + 30^{-1} (2 +30^{-1})\right) }}
 \leq \frac{0.18}{\kappa^2}  \leq \eps\sqrt{d},
\end{align*}}
where $\kappa\geq 1$ and $d\geq 1.$ Then ${\cal T}$ is a contraction mapping under~\eqref{def:casea}.

For Condition $(b)$ in~\eqref{def:caseb}, we have $\eps \sqrt{d} = 7\kappa^2\sqrt{d} \SNR^{-1}$, $\SNR\geq 35\kappa^4\sqrt{d}$, and thus $\eps^2 d \leq 1.4\sqrt{d}\SNR^{-1}. $ Then 
\begin{align*}
\rho & \leq  \frac{2\kappa^2\left(14\kappa^2\sqrt{d} + 2+ \SNR^{-1}  \right)\SNR^{-1} }{ 1 - \left(0.7\sqrt{d} +  \kappa(2 + \SNR^{-1})\right)\SNR^{-1} }\leq 1, \\
\eta  & \leq \frac{4\kappa^2\sqrt{d} \left(0.7  + \kappa^{-1}(2+\SNR^{-1})\right)\SNR^{-1} }{  2 - \left( 0.7\sqrt{d} + \kappa (2 + \SNR^{-1})\right)\SNR^{-1} }  \leq 6\kappa^2 \sqrt{d}\SNR^{-1}  \leq \eps\sqrt{d},
\end{align*}
for $\SNR\geq 35\kappa^4\sqrt{d}.$ 
\end{proof}

The following proposition guarantees that the spectral initialization in Algorithm~\ref{algo1} will produce an initial guess $\BO^0$ in $\mathcal{N}_{\eps}$.
\begin{proposition}[\bf Spectral initialization]\label{prop:init}
If the SNR in~\eqref{def:snr} satisfies $\SNR > 1$, then
\[
d_F(\BO^0,\BO) = \min_{\BQ\in\Od(d)}\|\BO^0 - \BO\BQ\|_F \leq \frac{6\kappa^2 (\SNR+1)^2}{(\SNR-1)^2}  \SNR^{-1}\sqrt{nd}.
\]
In addition, if $\SNR \geq 35\kappa^4\sqrt{d}$ and $\eps = 7\kappa^2\SNR^{-1}$, i.e., Condition ($b$) in~\eqref{def:caseb} holds, then
\[
d_F(\BO^0,\BO) \leq 7\kappa^2 \SNR^{-1}{\sqrt{nd}},
\]
the spectral initialization produces $\BO^0\in\mathcal{N}_{\eps}.$ 
\end{proposition}

\subsection{${\cal T}$ is a contraction mapping on $\mathcal{N}_{\eps}$}

We start with providing a few supporting lemmas that will be used to prove Proposition~\ref{prop:con} and~\ref{prop:basin}.
Recall that
\begin{equation}\label{def:bardelta}
\BC = \BO\BA\BA^{\top}\BO^{\top} + \bar{\BDelta},~~\bar{\BDelta} := \BDelta\BA^{\top}\BO^{\top} + \BO\BA\BDelta^{\top} + \BDelta\BDelta^{\top}.
\end{equation}

\begin{lemma}\label{lem:sigmamin}
The operator norm of $\bar{\BDelta}$ is bounded by
\begin{equation}\label{eq:bardelta}
\|\bar{\BDelta}\| \leq n(2+\SNR^{-1})\|\BA\|\max\|\BDelta_i\| = n \|\BA\|\sigma_{\min}(\BA)(2+\SNR^{-1})\SNR^{-1}.
\end{equation}
For any $\BX\in \mathcal{N}_{\eps}$, it holds for $1\leq i\leq n$,
\begin{align*}
\max\|\bar{\BDelta}_i^{\top}\BX\| & \leq n\|\BA\| \sigma_{\min}(\BA) (2+\SNR^{-1})\SNR^{-1}, \\
\sigma_{\max}([\BC\BX]_i) & \leq n\|\BA\|^2\left(1 + \kappa^{-1}(2+\SNR^{-1})\SNR^{-1}\right),  \\
\sigma_{\min}( [\BC\BX]_i ) & \geq n\|\BA\|^2\kappa^{-2} \left( 1 - \eps^2d/2 - \kappa (2+\SNR^{-1})\SNR^{-1}\right),
\end{align*}
where $\bar{\BDelta}_i = \BDelta(\BO_i\BA +  \BDelta_i)^{\top} + \BO\BA\BDelta_i^{\top}$ is the $i$-th block column of $\bar{\BDelta}.$
\end{lemma}

{Lemma~\ref{lem:sigmamin} is essential and also is different from~\cite{L23d}. Here we do not assume any statistical property of noise and thus use linear algebra to have a deterministic bound on $\|\bar{\BDelta}\BX\|_F$. It could be suboptimal under many statistical scenarios. On the other hand, the Gaussianity of noise was exploited in~\cite{L23d} and we can carry out a careful analysis  to control $\|\bar{\BDelta}\BX^t\|_F$ via the leave-one-out technique where $\BX^t$ is the iterate from the GPM.}

\begin{proof}[\bf Proof of Lemma~\ref{lem:sigmamin}]
The operator norm of $\bar{\BDelta}$ satisfies
\begin{align*}
\|\bar{\BDelta}\| & \leq 2\sqrt{n}\|\BA\|\| \BDelta\|  + \|\BDelta\|^2 = \sqrt{n}\|\BA\|\|\BDelta\|\left(2 + \frac{\|\BDelta\|}{\sqrt{n}\|\BA\|}\right) \leq n(2+\SNR^{-1})\|\BA\|\max\|\BDelta_i\| 
\end{align*}
where $\|\BDelta\|\leq \sqrt{n}\max\|\BDelta_i\| = \sqrt{n}\sigma_{\min}(\BA) \SNR^{-1}\leq \sqrt{n}\|\BA\|\SNR^{-1}.$

We consider an upper bound of $[\BC \BX]_i = \BO_i\BA\BA^{\top}\BO^{\top}\BX + \bar{\BDelta}_i^{\top}\BX$:
\begin{align}
\|\bar{\BDelta}_i^{\top}\BX\| 
& =   \|\BDelta_i\BA^{\top}\BO^{\top}\BX + (\BO_i\BA +  \BDelta_i) \BDelta^{\top}\BX\| \nonumber \\
& \leq  n\|\BA\| \|\BDelta_i\|+ n(\|\BA\| + \|\BDelta_i \|) \max\|\BDelta_i\|  \nonumber  \\
& \leq n\|\BA\| \sigma_{\min}(\BA) (2 + \SNR^{-1})\SNR^{-1} \label{eq:DeltaS}
\end{align}
where $\|\BO^{\top}\BX\| \leq n$ and $\|\BDelta^{\top}\BX\| \leq \sqrt{n}\|\BDelta\| \leq n\max\|\BDelta_i\|.$

For the lower bound of $\sigma_{\min}([\BC\BX]_i)$, it holds that
\begin{align*}
\sigma_{\min}([\BC\BX]_i) & \geq \sigma_{\min}(\BO_i\BA\BA^{\top}\BO^{\top}\BX) - \|\bar{\BDelta}_i^{\top}\BX\| \geq  \sigma_{\min}^2(\BA) \sigma_{\min}(\BO^{\top}\BX) - \|\bar{\BDelta}_i^{\top}\BX\|\\
& \geq  \sigma_{\min}^2(\BA) \left( 1 - \frac{\eps^2d}{2}\right) - n\|\BA\| \sigma_{\min}(\BA) (2 + \SNR^{-1})\SNR^{-1} \\
& \geq n\|\BA\|^2 \kappa^{-2} \left( 1 - \eps^2d/2 - \kappa (2+\SNR^{-1})\SNR^{-1}\right)
\end{align*}
where the singular values of $\BO^{\top}\BX$ satisfy~\eqref{eq:sigmaZS}.
Similarly, we have
\begin{align*}
\sigma_{\max}([\BC\BX]_i) & \leq  \|\BA\BA^{\top}\BO^{\top}\BX\| + \|\bar{\BDelta}_i^{\top}\BX \|  \leq n\|\BA\|^2 + n\|\BA\| \sigma_{\min}(\BA) (2 + \SNR^{-1})\SNR^{-1} \\
&  \leq n\|\BA\|^2\left(1 + \kappa^{-1}(2+\SNR^{-1})\SNR^{-1}\right).
\end{align*}
\end{proof}

\begin{lemma}\label{lem:conL}
For any $\BX$ and $\BY$ in $\mathcal{N}_{\eps}\subseteq\Od(d)^{\otimes n}$, we have
\begin{align*}
d_F(\BC\BX, \BC\BY) & \leq n\|\BA\|^2\left( 2\eps\sqrt{d} + (2+ \SNR^{-1}) \SNR^{-1} \right) d_F(\BX,\BY).
\end{align*}

\end{lemma}
\begin{proof}
Let $\BQ$ be the orthogonal matrix which makes $d_F(\BX,\BY) = \|\BX - \BY\BQ\|_F$, i.e., $\BQ=\PP(\BY^{\top}\BX).$
\begin{align*}
\|\BC(\BX-\BY\BQ)\|_F 
& \leq  \| \BO\BA\BA^{\top}\BO^{\top}(\BX - \BY\BQ) \|_F +  \|\bar{\BDelta}(\BX - \BY\BQ)\|_F \\
& \leq \sqrt{n}\|\BA\|^2 \|\BO^{\top}(\BX - \BY\BQ)\|_F + \|\bar{\BDelta}\| \|\BX - \BY\BQ\|_F \\
& \leq \sqrt{n}\|\BA\|^2 \|\BO^{\top}(\BX - \BY\BQ)\|_F + n(2+ \SNR^{-1})\|\BA\|\max\|\BDelta_i\|  \|\BX - \BY\BQ\|_F \\
& \leq n\|\BA\|^2 \left( \frac{1}{\sqrt{n}} \|\BO^{\top}(\BX - \BY\BQ)\|_F + (2+ \SNR^{-1}) \SNR^{-1} \|\BX - \BY\BQ\|_F\right) 
\end{align*}
where  $\|\bar{\BDelta}\| \leq n(2+ \SNR^{-1})\|\BA\|\max\|\BDelta_i\|$ in~\eqref{eq:bardelta}.
The first term is bounded by
\begin{align*}
\|\BO^{\top}(\BX - \BY\BQ)\|_F & \leq \| (\BY - \BO\BQ_Y)^{\top}(\BX - \BY\BQ) \|_F + \|\BY^{\top}(\BX - \BY\BQ)\|_F \\
\text{Use }\BQ_Y = \PP(\BO^{\top}\BY)~~\qquad & \leq \eps\sqrt{nd}\cdot d_F(\BX,\BY) + \|\BY^{\top}\BX - n\BQ\|_F \\
\text{Use}~\eqref{eq:YX}~~\qquad  & \leq \eps\sqrt{nd}\cdot d_F(\BX,\BY) + \frac{1}{2}\|\BX - \BY\BQ\|_F^2 \\
& \leq 2\eps\sqrt{nd}\cdot d_F(\BX,\BY)
\end{align*}
where $d_F(\BX,\BY) = \min_{\BQ\in\Od(d)}\|\BX- \BY\BQ\|_F \leq d_F(\BX,\BO) + d_F(\BY,\BO) \leq 2\eps\sqrt{nd}$ and 
\begin{equation}\label{eq:YX}
\| \BY^{\top} \BX - n\BQ\|_F \leq \frac{1}{2}\|\BX- \BY\BQ\|_F^2
\end{equation}
for any $\BX$ and $\BY\in\Od(d)^{\otimes n}$ where $\BQ : = \PP(\BY^{\top}\BX)$. Thus
\begin{align*}
d_F(\BC\BX, \BC\BY) & \leq \|\BC(\BX-\BY\BQ)\|_F \leq n\|\BA\|^2\left( 2\eps\sqrt{d} +(2+ \SNR^{-1}) \SNR^{-1} \right) d_F(\BX,\BY).
\end{align*}
The inequality~\eqref{eq:YX} follows from
\begin{align*}
\| \BY^{\top} \BX - n\BQ \|_F^2 & = \sum_{i=1}^d(n - \sigma_i(\BY^{\top}\BX))^2  \leq \left( \sum_{i=1}^d(n - \sigma_i(\BY^{\top}\BX)) \right)^2 \\
\text{Use }\eqref{def:df2} ~~\qquad & = \left(nd - \|\BY^{\top}\BX\|_*\right)^2 = \left( \frac{1}{2}\|\BX- \BY\BQ\|_F^2 \right)^2
\end{align*}
where $\sigma_i(\BY^{\top}\BX)\leq n$.
\end{proof}

\begin{lemma}\label{lem:conP}
For two invertible matrices $\BX$ and $\BY$ in $\RR^{d\times d}$, 
\[
\|\PP(\BX) - \PP(\BY)\| \leq \frac{2\|\BX - \BY\|}{\sigma_{\min}(\BX)+\sigma_{\min}(\BY)},~~~
\|\PP(\BX) - \PP(\BY)\|_F \leq \frac{4\|\BX - \BY\|_F}{\sigma_{\min}(\BX)+\sigma_{\min}(\BY)}
\]
{where $\PP(\cdot)$ is defined in~\eqref{def:proj}. }
\end{lemma}
The proof uses Davis-Kahan theorem for eigenvector perturbation~\cite{DK70}.
This perturbation bound can be found in~\cite[Theorem 1]{L95} and~\cite[Lemma 4.7]{L20c}.
It seems that to prove ${\cal T}$ is a contraction, we only need to combine Lemma~\ref{lem:conL} with Lemma~\ref{lem:conP}. 
The proof of Proposition~\ref{prop:con} follows from combining all the results above.
\begin{proof}[\bf Proof of Proposition~\ref{prop:con}]
For any $\BX$ and $\BY$ in $\mathcal{N}_{\eps}$, Lemma~\ref{lem:sigmamin} implies a lower bound of $\sigma_{\min}([\BC\BX]_i)$ and $\sigma_{\min}([\BC\BY]_i)$. Let $\BQ$ be the orthogonal matrix which minimizes $\|\BC\BX - \BC\BY\BQ\|_F$. We have 
\begin{align*}
d_F(\PP_n(\BC\BX), \PP_n(\BC\BY)) & {= \min_{\BR\in\Od(d)}\| \PP_n(\BC\BX) - \PP_n(\BC\BY)\BR  \|_F} \\
& \leq \sqrt{\sum_{i=1}^n \| \PP([\BC\BX]_i) - \PP([\BC\BY]_i)\BQ \|_F^2} \\
\text{Lemma}~\ref{lem:conP}~~ \qquad 
& \leq 4\sqrt{\sum_{i=1}^n \frac{\| [\BC\BX]_i - [\BC\BY]_i\BQ \|_F^2}{( \sigma_{\min}([\BC\BX]_i) + \sigma_{\min}([\BC\BY]_i) )^2} } \\
{\text{Lemma}~\ref{lem:sigmamin}\text{ and }~\ref{lem:conL}}~~\qquad & \leq \frac{1}{n\|\BA\|^2}\cdot \frac{2\kappa^2}{  1 - \eps^2d/2 -  \kappa(2 + \SNR^{-1})\SNR^{-1} }\cdot d_F(\BC\BX, \BC\BY) \\
& \leq\frac{2\kappa^2\left(2\eps\sqrt{d} + (2+ \SNR^{-1}) \SNR^{-1} \right) }{ 1 - \eps^2d/2 -  \kappa(2 + \SNR^{-1})\SNR^{-1} }\cdot d_F(\BX, \BY) 
\end{align*} 
which gives the expression of $\rho.$  
\end{proof}

\begin{proof}[\bf Proof of Proposition~\ref{prop:basin}]
Suppose $\BX\in\mathcal{N}_{\eps}$, then we want to show that $\PP_n(\BC\BX)\in\mathcal{N}_{\eps}$. 
Let $\BQ = \PP(\BO^{\top}\BX)$ be the minimizer to $\min_{\BQ\in\Od(d)}\|\BX - \BO\BQ\|_F$ {whose optimal value equals $d_F(\BX,\BO)$ in~\eqref{def:df2} exactly.} {Then we try to bound the distance between ${\cal T}(\BX)$ and $\BO$ via}
\begin{align}
d_F({\cal T}(\BX), \BO) & =  d_F(\PP_n(\BC\BX), \PP_n(n\BO\BA\BA^{\top})) \nonumber \\
& \leq \sqrt{ \sum_{i=1}^n \| \PP([\BC\BX]_i) - \PP(n\BO_i\BA\BA^{\top})\BQ \|_F^2}  \nonumber \\
& \leq \frac{1}{n\|\BA\|^2} \cdot \frac{4\kappa^2}{  2 - \left( \eps^2d/2 + \kappa (2 + \SNR^{-1})\SNR^{-1}\right) } \cdot \left\| \BC\BX - n\BO\BA\BA^{\top}\BQ \right\|_F \label{eq:PZ}
\end{align}
where $\sigma_{\min}(n\BO_i\BA\BA^{\top}) \geq n \|\BA\|^2\kappa^{-2}$, { ${\cal P}(n\BO_i\BA\BA^{\top}) = \BO_i {\cal P}(\BA\BA^{\top}) = \BO_i$, and thus ${\cal P}_n(n\BO\BA\BA^{\top}) = \BO.$}
{ It suffices to estimate $\|\BC\BX - n\BO\BA\BA^{\top}\BQ\|_F$. Note that for each block of $\BC\BX - n\BO\BA\BA^{\top}\BQ$, it holds that}
\begin{align*}
\left\| [\BC\BX]_i -  {n\BO_i\BA\BA^{\top}\BQ} \right\|_F
& = \left\| (\BO_i \BA\BA^{\top}\BO^{\top} \BX + \bar{\BDelta}_i^{\top}\BX)  - n \BO_i\BA\BA^{\top}\BQ\right\|_F \\
& \leq   \| \BO_i\BA\BA^{\top} (\BO^{\top}\BX - n\BQ)\|_F + \| \bar{\BDelta}_i^{\top}\BX \|_F \\
\text{Use}~\eqref{eq:DeltaS} \qquad & \leq \|\BA\|^2 \|\BO^{\top}\BX - n\BQ\|_F + n\sqrt{d}\|\BA\| \sigma_{\min}(\BA) (2+\SNR^{-1})\SNR^{-1}  \\
\text{Use}~\eqref{eq:YX} \qquad & \leq \frac{\|\BA\|^2\|\BX - \BO\BQ\|_F^2}{2}  + n\|\BA\|^2 \kappa^{-1}\sqrt{d} (2 + \SNR^{-1})\SNR^{-1} \\
\text{Use}~\BX\in{\cal N}_{\eps}\qquad  & \leq n\|\BA\|^2\left(\frac{\eps^2 d}{2}  + \kappa^{-1}\sqrt{d}(2+\SNR^{-1})\SNR^{-1}\right).
\end{align*}
By combining it with~\eqref{eq:PZ}, we have
\begin{align*}
d_F({\cal T}(\BX), \BO) & \leq \frac{4\sqrt{n}\kappa^2\left(\eps^2 d/2  + \kappa^{-1}\sqrt{d}(2+\SNR^{-1})\SNR^{-1}\right) }{  2 - \left( \eps^2d/2 + \kappa (2 + \SNR^{-1})\SNR^{-1}\right) } \leq \eps\sqrt{nd}
\end{align*}
provided that~\eqref{cond:basin} holds.
\end{proof}

\subsection{Proof of Theorem~\ref{thm:main},~\ref{thm:gpm}, and~\ref{thm:bm}}

Step 1-3 are devoted to proving Theorem~\ref{thm:main} and~\ref{thm:gpm}; and Step 4 justifies Theorem~\ref{thm:bm}.

\paragraph{\bf Step 1: Existence of the fixed point and linear convergence.}
Let $\mathcal{N}_{\eps} = \{ \BX\in \Od(d)^{\otimes n}: d_F(\BX,\BO) \leq \eps\sqrt{nd} \}$ and we equip it with the distance function $d_F(\cdot,\cdot)$ in~\eqref{def:df}. Apparently, this is a complete and compact metric space since $\mathcal{N}_{\eps}$ is a bounded and closed subset in $\RR^{nd\times d}$. 
For the operator ${\cal T}$ defined in~\eqref{def:T}, Proposition~\ref{prop:key} ensures that ${\cal T}$ is a contraction mapping on $\mathcal{N}_{\eps}$ and it maps $\mathcal{N}_{\eps}$ to itself under~\eqref{def:casea} and~\eqref{def:caseb}. 
\vskip0.2cm
For the proof of Theorem~\ref{thm:main} and~\ref{thm:gpm}, we use different initializations. 
\begin{itemize}
\item To prove the tightness of the SDR in Theorem~\ref{thm:main}, we simply set $\BO^0 = \BO\in\mathcal{N}_{\eps}$; 
\item To prove the global convergence of the GPM via spectral initialization in Theorem~\ref{thm:gpm}, Proposition~\ref{prop:init} along with Condition $(b)$ in~\eqref{def:caseb} implies that the initial guess $\BO^0$ satisfies $d_F(\BO^0,\BO) \leq \eps\sqrt{nd}$ with $\eps = 7\kappa^2\SNR^{-1}$ and $\SNR\geq 35\kappa^4\sqrt{d}.$
\end{itemize}

Theorem~\ref{thm:banach} guarantees that the generalized power method with an initial value $\BO^0\in \mathcal{N}_{\eps}$ will finally converge to a unique fixed point $\widehat{\BO}$ in $\mathcal{N}_{\eps}$ satisfying
\[
d_F({\cal T}(\widehat{\BO}), \widehat{\BO}) = 0 ~~\Longleftrightarrow~~{\cal T}(\widehat{\BO}) = \PP_n(\BC\widehat{\BO})= \widehat{\BO}\BQ
\]
for some $\BQ\in\Od(d)$.
Moreover, the convergence rate satisfies $d_F(\BO^t, \widehat{\BO}) \leq \rho^t d_F(\BO^0, \widehat{\BO})$
where $\rho<1$ holds under~\eqref{def:casea} and~\eqref{def:caseb} according to Proposition~\ref{prop:key}.

\vskip0.25cm

\paragraph{\bf Step 2: Proof of ${\cal T}(\widehat{\BO})= \widehat{\BO}$.}
Next, we will show that $\BQ$ is in fact equal to $\I_d.$ Note that Lemma~\ref{lem:sigmamin} implies $\sigma_{\min}([\BC\widehat{\BO}]_i) > 0$ for $\widehat{\BO}\in \mathcal{N}_{\eps}$ with $\eps$ satisfying~\eqref{def:casea} or~\eqref{def:caseb}.
For each block of $\BC\widehat{\BO}$, we get
\[
\PP([\BC\widehat{\BO}]_i) = ([\BC\widehat{\BO}]_i[\BC\widehat{\BO}]_i^{\top})^{-1/2}[\BC\widehat{\BO}]_i = \widehat{\BO}_i\BQ,~1\leq i\leq n.
\]
Define
\begin{equation}\label{def:blambda}
\BLambda_{ii} = ([\BC\widehat{\BO}]_i[\BC\widehat{\BO}]_i^{\top})^{1/2},~~\BLambda = \blkdiag(\BLambda_{11},\cdots,\BLambda_{nn})
\end{equation}
which is a strictly positive semidefinite block-diagonal matrix,
and then
\begin{equation}\label{eq:fixCS}
[\BC\widehat{\BO}]_i = \BLambda_{ii} \widehat{\BO}_i\BQ,~1\leq i\leq n ~~\Longleftrightarrow~~ \BC\widehat{\BO} = \BLambda\widehat{\BO}\BQ.
\end{equation}

Now multiplying both sides of~\eqref{eq:fixCS}  by $\widehat{\BO}^{\top}$ gives
\begin{align}
\widehat{\BO}^{\top}\BC\widehat{\BO} = \widehat{\BO}^{\top}\BO\BA\BA^{\top}\BO^{\top} \widehat{\BO}+  \widehat{\BO}^{\top}\bar{\BDelta}\widehat{\BO} = \widehat{\BO}^{\top}\BLambda\widehat{\BO}\BQ,  \label{eq:posQ}
\end{align}
where $\BC$ satisfies~\eqref{def:bardelta}.
The smallest eigenvalue of $\widehat{\BO}^{\top}\BC\widehat{\BO}$ is lower bounded by
\begin{align*}
\lambda_{\min}(\widehat{\BO}^{\top}\BC\widehat{\BO})
& \geq \sigma_{\min}^2(\BA) \cdot \lambda_{\min}( \widehat{\BO}^{\top}\BO\BO^{\top} \widehat{\BO}) - n\|\bar{\BDelta}\| \\
\text{Use}~\eqref{eq:bardelta} ~~\qquad & \geq \sigma_{\min}^2(\BA) \cdot \sigma_{\min}^2 (\BO^{\top}\widehat{\BO}) -3n^{2}\|\BA\|\max\|\BDelta_i\| \\
\text{Use}~\eqref{eq:sigmaZS}~~\qquad & \geq   \sigma_{\min}^2(\BA) \left( 1- \eps^2 d/2\right)^2n^2 - 3n^{2}\|\BA\|\max\|\BDelta_i\| \\
& \geq n^2\sigma^2_{\min}(\BA) \left( (1-\eps^2d/2)^2  - 3\kappa \SNR^{-1}\right) > 0
\end{align*}
under both~\eqref{def:casea} and~\eqref{def:caseb}.
Note $\widehat{\BO}^{\top}\BC\widehat{\BO}$ is positive definite and so is $\widehat{\BO}^{\top}\BLambda\widehat{\BO}$. Therefore, $\BQ$ must be $\I_d$ according to~\eqref{eq:posQ} since the only positive semidefinite orthogonal matrix is $\I_d$. Therefore, $\widehat{\BO}$ is the only fixed point to the nonlinear mapping ${\cal T}$, i.e., ${\cal T}(\widehat{\BO}) = \widehat{\BO}$ and $(\BLambda - \BC)\widehat{\BO} = 0$ where $\BLambda$ is defined in~\eqref{def:blambda}.

\vskip0.25cm

\paragraph{\bf Step 3: Global optimality of $\widehat{\BO} \widehat{\BO}^{\top}$ in~\eqref{def:sdr}.}
Define $\BL = \BLambda - \BC$ where $\BLambda_{ii}$ is defined  in~\eqref{def:blambda} and $\widehat{\BQ} : = \argmin_{\BQ\in\Od(d)}\|\widehat{\BO} - \BO\BQ\|_F.,$ and we have $\BL\widehat{\BO} = 0.$
To prove the global optimality of $\widehat{\BO}\widehat{\BO}^{\top}$, it suffices to show $\lambda_{d+1}(\BL) > 0$, i.e., the $(d+1)$th smallest eigenvalue of $\BL$ is strictly positive since $\BL\widehat{\BO} = 0$ implies $d$ eigenvalues are equal to 0. Besides the $(d+1)$-th smallest eigenvalue, we will also consider the top eigenvalue of $\BL$ as it will be used in analyzing the optimization landscape.
{Note that
\begin{align*}
\BL & =  \left(\I_{nd} -n^{-1} \widehat{\BO}\widehat{\BO}^{\top}\right)\BL \left(\I_{nd} - n^{-1}\widehat{\BO}\widehat{\BO}^{\top}\right) \\
& = \left(\I_{nd} -n^{-1} \widehat{\BO}\widehat{\BO}^{\top}\right)\BLambda \left(\I_{nd} - n^{-1}\widehat{\BO}\widehat{\BO}^{\top}\right) - \left(\I_{nd} -n^{-1} \widehat{\BO}\widehat{\BO}^{\top}\right)\BC \left(\I_{nd} - n^{-1}\widehat{\BO}\widehat{\BO}^{\top}\right) 
\end{align*}
where $\BL\widehat{\BO} = 0.$ As a result, the Weyl's theorem on eigenvalues implies}
the $(d+1)$-th smallest and top eigenvalues of $\BL$ are 
\begin{align*}
\lambda_{\max}(\BL) & \leq \lambda_{\max}(\BLambda) + \|  (\I_{nd} - n^{-1}\widehat{\BO}\widehat{\BO}^{\top} )  \BC  (\I_{nd} - n^{-1}\widehat{\BO}\widehat{\BO}^{\top} ) \|, \\
\lambda_{d+1}(\BL) & \geq \lambda_{\min}(\BLambda) - \|  (\I_{nd} - n^{-1}\widehat{\BO}\widehat{\BO}^{\top} )  \BC  (\I_{nd} - n^{-1}\widehat{\BO}\widehat{\BO}^{\top} ) \|.
\end{align*}
Note that $\BC = \BO\BA\BA^{\top}\BO^{\top} + \bar{\BDelta}$ and
\begin{align*}
& \|  (\I_{nd} - n^{-1}\widehat{\BO}\widehat{\BO}^{\top} )  \BC  (\I_{nd} - n^{-1}\widehat{\BO}\widehat{\BO}^{\top} ) \|
\\
&\qquad \leq \| (\I_{nd} - n^{-1}\widehat{\BO}\widehat{\BO}^{\top})\BO\BA\|_F^2 + \|\bar{\BDelta}\|  \leq \|\BA\|^2 \| (\I_{nd} - n^{-1}\widehat{\BO}\widehat{\BO}^{\top})\BO\|_F^2 + \|\bar{\BDelta}\|  \\
& \qquad \leq \|\BA\|^2 \|\widehat{\BO}-\BO\BQ\|_F^2 + \|\bar{\BDelta}\|  \leq  \eps^2 nd\|\BA\|^2 + n \|\BA\|\sigma_{\min}(\BA)(2+\SNR^{-1})\SNR^{-1}  \\
& \qquad \leq n\|\BA\|^2 \left(\eps^2 d + (2+\SNR^{-1})\SNR^{-1}\right)
\end{align*}
where $\|\bar{\BDelta}\|$ is bounded in~\eqref{eq:bardelta} and
\begin{equation}\label{eq:ohato}
 \| (\I_{nd} - n^{-1}\widehat{\BO}\widehat{\BO}^{\top})\BO\|_F =  \| (\I_{nd} - n^{-1}\widehat{\BO}\widehat{\BO}^{\top})(\BO - \widehat{\BO} \widehat{\BQ}^{\top})\|_F \leq \|\widehat{\BO} - \BO\widehat{\BQ}\|_F \leq \eps\sqrt{nd}.
\end{equation}

Note that $\BLambda_{ii}$ is symmetric and positive semidefinite, and thus its eigenvalues match the singular values of $[\BC\widehat{\BO}]_i$, i.e.,
$[\BC\widehat{\BO}]_i = \BO_i\BA\BA^{\top} \BO^{\top}\widehat{\BO} + \bar{\BDelta}_i^{\top}\widehat{\BO}$
where $\bar{\BDelta}_i$ is the $i$-th block column of $\bar{\BDelta}.$
Thus Lemma~\ref{lem:sigmamin} gives
\begin{align*}
\lambda_{\max}(\BLambda) 
& = \max_{1\leq i\leq n} \sigma_{\max}([\BC\widehat{\BO}]_i) \leq n\|\BA\|^2(1 + \kappa^{-1}(2+\SNR^{-1})\SNR^{-1}), \\
\lambda_{\min}(\BLambda) & = \min_{1\leq i\leq n} \sigma_{\min}([\BC\widehat{\BO}]_i) \geq n\sigma_{\min}^2(\BA)\left( 1 - \eps^2d/2 - \kappa (2+\SNR^{-1})\SNR^{-1}\right).
\end{align*}

Then we have 
\begin{align*}
\lambda_{\max}(\BL) &\leq n\|\BA\|^2 +  n\|\BA\|^2 \left(\eps^2 d + 2(2+\SNR^{-1})\SNR^{-1}\right).
\end{align*}
For $\lambda_{d+1}(\BL)$, it holds that
\begin{align*}
\lambda_{d+1}(\BL) & \geq n\sigma_{\min}^2(\BA)\left( 1 - \eps^2d/2 - \kappa (2+\SNR^{-1})\SNR^{-1}\right) - n\|\BA\|^2 \left(\eps^2 d + (2+\SNR^{-1})\SNR^{-1}\right) \\
& \geq n\sigma_{\min}^2(\BA)\left( 1 - 3\kappa^2\eps^2d/2 - 2\kappa^2 (2+\SNR^{-1})\SNR^{-1}\right) > 0
\end{align*}
under both~\eqref{def:casea} and~\eqref{def:caseb}.
Theorem~\ref{thm:cvx} implies that $\widehat{\BO}\widehat{\BO}^{\top}$ is the unique global minimizer to the SDR, which finishes the proof of Theorem~\ref{thm:main} and~\ref{thm:gpm}.

\paragraph{\bf Step 4: Optimization landscape of the Burer-Monteiro factorization.}

Under condition ($b$) in~\eqref{def:caseb}, we have $\eps^2d \leq 1.4 \sqrt{d}\SNR^{-1}$ for $\SNR\geq 35\kappa^4\sqrt{d}$ and $\eps = 7\kappa^2\SNR^{-1}$. The eigenvalues of $\BL$ satisfy
\begin{align*}
\lambda_{\max}(\BL) & \leq n\|\BA\|^2 \left(1 +\eps^2 d + 2(2+\SNR^{-1})\SNR^{-1}\right) \leq n\|\BA\|^2(1 + 6\sqrt{d}\SNR^{-1}), \\
\lambda_{d+1}(\BL) & \geq n\sigma_{\min}^2(\BA)\left( 1 - 3\kappa^2\eps^2d/2 - 2\kappa^2 (2+\SNR^{-1})\SNR^{-1}\right) \geq n\sigma_{\min}^2(\BA)(1 - 7\kappa^2\sqrt{d}\SNR^{-1}).
\end{align*}
As a result, we have
\[
\frac{\lambda_{\max}(\BL)}{\lambda_{d+1}(\BL)} \leq \kappa^2\cdot \frac{1+ 7\kappa^2\sqrt{d}\SNR^{-1}}{1-7\kappa^2\sqrt{d}\SNR^{-1}}.
\]
Theorem 2.8 in~\cite{L23c} implies that if $p\geq (2\lambda_{\max}(\BL)/\lambda_{d+1}(\BL) - 1)d + 2$, i.e.,
\[
p \geq \left(\frac{2\kappa^2 (1+ 7\kappa^2\sqrt{d}\SNR^{-1})}{1-7\kappa^2\sqrt{d}\SNR^{-1}} -1 \right)d + 2~ \Longleftrightarrow~ \SNR\geq 7\kappa^2\sqrt{d} \cdot \frac{p+2\kappa^2 d + d-2}{p-2\kappa^2 d+d-2},
\]
then the optimization landscape of~\eqref{def:bm} is benign, i.e., every local minimizer is global and the local minimizer $\BS\in\St(d,p)^{\otimes n}$ satisfies $\BS\BS^{\top} = \widehat{\BO}\widehat{\BO}^{\top}.$ This finishes the proof of Theorem~\ref{thm:bm}.

\subsection{Proof of Proposition~\ref{prop:init} (Initialization)}\label{sss:init}

The proof of Proposition~\ref{prop:init} relies on the Davis-Kahan theorem for singular vectors~\cite{DK70,W72}. 
\begin{theorem}[\bf Davis-Kahan theorem]
Let $\BB$ and $\widetilde{\BB}$ be two matrices of size $\RR^{N\times M}$ with $\min\{N,M\}\geq d$.
Suppose 
\[
\sigma_{d}(\BB) \geq \alpha + \delta, \qquad \sigma_{d+1}(\widetilde{\BB})\leq \alpha, 
\]
where $\sigma_d(\cdot)$ denotes the $d$th largest singular value,
then
\[
\frac{1}{2}\min_{\BQ\in\Od(d)}\|\BV_{\widetilde{\BB}} - \BV_{\BB}\BQ\|\leq \left\|  (\I_N - \BV_{\BB}\BV_{\BB}^{\top})\BV_{\widetilde{\BB}}   \right\| \leq \frac{\|\widetilde{\BB} - \BB\|}{\delta}
\]
where $\BV_{\BB}$ and $\BV_{\widetilde{\BB}}\in\RR^{N\times d}$ are the top $d$ normalized left singular vectors of $\BB$ and $\widetilde{\BB}$ respectively. A similar bound also holds for the right singular vectors.
\end{theorem}

Now we present the proof of Proposition~\ref{prop:init}.
\begin{proof}[\bf Proof of Proposition~\ref{prop:init}]
Consider the data matrix $\BD := \BO\BA + \BDelta\in\RR^{nd\times m}$. 
 Denote $(\BSigma,\BU, \BV)$ as the top $d$ singular values and left/right singular vectors of $\BD$ where $\BSigma\in\RR^{d\times d}$ is a diagonal matrix, $\BU\in\RR^{nd\times d}$ with $\BU^{\top}\BU=\I_d$, and $\BV\in\RR^{m\times d}$ with $\BV^{\top}\BV=\I_d$. It simply holds
\[
 \BU = \BD\BV\BSigma^{-1}.
\]
Let $\BA= \bar{\BU} \bar{\BSigma} \bar{\BV}^{\top}\in\RR^{d\times m}$ be the SVD of $\BA$ and $\bar{\BV}\in\RR^{m\times d}$ consists of top $d$ right singular vectors of $\BA.$ We will approximate $\BU$ by using ``one-step power method":
\[
\BU - \BD \bar{\BV}\BQ \BSigma^{-1} = \BD(\BV - \bar{\BV}\BQ)\BSigma^{-1}
\]
where $\BQ$ is an orthogonal matrix that minimizes $\min_{\BR\in\Od(d)}\|\BV - \bar{\BV}\BR\|_F$. Next, we will show that each block $\BU_i$ of $\BU$ is well approximated by that of $\BD\bar{\BV}\BQ\BSigma^{-1}$. We start with estimating $\|\BV-\bar{\BV}\BQ\|$ using the Davis-Kahan theorem.

Using Weyl's inequality, we have
\begin{align}
\| \BD - \BO\BA\| \leq \|\BDelta\| \quad &  \Longrightarrow \quad |\sigma_i(\BD) - \sqrt{n}\sigma_i(\BA)| \leq \|\BDelta\|,~~1\leq i\leq d \label{eq:WDZ}.
\end{align}
By denoting $\BSigma = \diag(\sigma_1(\BD),\cdots,\sigma_d(\BD))$, it holds that
\[
\sqrt{n}\|\BA\|(1 + \SNR^{-1})\geq \lambda_{\max}(\BSigma)\geq \lambda_{\min}(\BSigma) \geq \sqrt{n}\sigma_{\min}(\BA)(1 - \SNR^{-1}).
\]
and $\|\BDelta_i\| \leq \SNR^{-1} \sigma_{\min}(\BA) \leq \SNR^{-1}\|\BA\|.$

Applying the Davis-Kahan theorem to $\BD$ and $\BO\BA$ implies that
\begin{equation}\label{eq:DKV}
\| \BV - \bar{\BV}\BQ\| \leq \frac{2\|\BDelta\|}{ \sqrt{n}\sigma_{\min}(\BA) - \|\BDelta\| } \leq \frac{2\max\|\BDelta_i\|}{\sigma_{\min}(\BA) - \max \|\BDelta_i\|} = \frac{2}{\SNR - 1}.
\end{equation}
As a result, the $i$th block $\BU_i$ of $\BU$ is approximated by $[\BD \bar{\BV}\BQ \BSigma^{-1}]_i$ with error bounded by
\begin{align*}
\|\BU_i - [\BD \bar{\BV}\BQ \BSigma^{-1}]_i\| 
& \leq \| (\BO_i\BA + \BDelta_i) (\BV - \bar{\BV}\BQ) \BSigma^{-1}\| \\
\text{Use}~\eqref{eq:WDZ} \text{ and}~\eqref{eq:DKV} \qquad 
& \leq \|\BA\|(1 + \SNR^{-1}) \cdot \frac{2}{\SNR-1} \cdot \frac{1}{\sqrt{n}\sigma_{\min}(\BA)(1-\SNR^{-1})} \\
& \leq \frac{2\kappa (\SNR+1)}{\sqrt{n}(\SNR-1)^2}.
\end{align*}
Therefore, 
\begin{align*}
\|\BU_i - \BO_i\BA\bar{\BV} \BQ \BSigma^{-1}\| & \leq \|\BU_i - (\BO_i\BA + \BDelta_i)\bar{\BV} \BQ \BSigma^{-1}\| + \| \BDelta_i\bar{\BV} \BQ \BSigma^{-1}\| \\
& \leq \frac{2\kappa (\SNR+1)}{\sqrt{n}(\SNR-1)^2} + \frac{\max\|\BDelta_i\|}{\sqrt{n}\sigma_{\min}(\BA)(1-\SNR^{-1})} \\
& \leq \frac{2\kappa (\SNR+1)}{\sqrt{n}(\SNR-1)^2} + \frac{\SNR-1}{\sqrt{n}(\SNR-1)^2} \leq \frac{3\kappa (\SNR+1)}{\sqrt{n}(\SNR-1)^2}.
\end{align*}

Now applying Lemma~\ref{lem:conP} gives:
\begin{align*}
\| \PP(\BU_i) - \PP( \BO_i\BA\bar{\BV} \BQ \BSigma^{-1}) \| 
& \leq \frac{2}{\sigma_{\min}({\BO_i\BA\bar{\BV}\BQ\BSigma^{-1}}) }\cdot \|\BU_i - \BO_i \BA\bar{\BV}\BQ\BSigma^{-1}\| \\
& \leq 2\kappa\sqrt{n}(1+\SNR^{-1}) \cdot \frac{3\kappa (\SNR+1)}{\sqrt{n}(\SNR-1)^2} = \frac{6\kappa^2 (\SNR+1)^2}{(\SNR-1)^2}  \SNR^{-1}
\end{align*}
where $\sigma_{\min}(\bar{\BSigma}) = \sigma_{\min}(\BA)$ and
\begin{align*}
\sigma_{\min}(\BO_i\BA\bar{\BV} \BQ \BSigma^{-1}) & = \sigma_{\min}(\bar{\BU}\bar{\BSigma} \BQ \BSigma^{-1})  \geq \sigma_{\min}(\bar{\BSigma})\sigma_{\min}(\BSigma^{-1}) \\
\text{Use}~\eqref{eq:WDZ}\qquad & \geq \frac{\sigma_{\min}(\BA)}{\sqrt{n}\|\BA\|(1+\SNR^{-1})}  = \frac{1}{\kappa\sqrt{n}(1+\SNR^{-1})}.
\end{align*}
Finally we arrive at 
\begin{align*}
d_F(\BO^0,\BO) & \leq \sqrt{ \sum_{i=1}^n \| \PP(\BU_i) - \PP(\BO_i\BA\bar{\BV} \BQ \BSigma^{-1}) \|_F^2  }   \leq \frac{6\kappa^2 (\SNR+1)^2}{(\SNR-1)^2}  \SNR^{-1}  \sqrt{nd}.
\end{align*}

\end{proof}

\subsection{Proof of Corollary~\ref{cor:binary}}\label{ss:corbinary}
Define $\delta : = m^{-1}d\log(n(m+d))$, as we assume $\delta\leq 1/16$ which is guaranteed if $m \geq 16d\log (n(m+d)).$
We start with the estimation of singular values of $\BA$. Note that $\BA\BA^{\top} = \sum_{j=1}^m \ba_j\ba_j^{\top}$ and $\E \BA\BA^{\top} = md^{-1}\I_d$. Then by~\cite[Example 6.21]{W19}, we have
\[
\left\|\frac{d}{m}\BA\BA^{\top} - \I_d\right\| \leq 2\gamma \delta + \sqrt{2\gamma\delta}
\]
with probability $1 - 2d(n(m+d))^{-\gamma}$. Setting $\gamma=2$ gives
\[
\frac{m}{4d}\leq \frac{m(1 - 2\sqrt{\delta} - 4\delta)}{d} \leq \lambda_{\min}(\BA\BA^{\top}) \leq \lambda_{\max}(\BA\BA^{\top}) \leq \frac{m(1 + 2\sqrt{\delta} + 4\delta)}{d} \leq \frac{7m}{4d}.
\]

It is straightforward to control $\|\BDelta_i\|$ via the Bernstein inequality (\cite[Theorem 1.4]{T12}): conditioned on $\BA$, it holds{
\begin{align*}
&   \sum_{j=1}^m\E \BDelta_{ij} \BDelta_{ij}^{\top} = \E \BDelta_i\BDelta_i^{\top}= \sum_{j=1}^m \theta(1-\theta) \ba_j\ba_j^{\top} + (1-\theta) \E \bz_{ij}\bz_{ij}^{\top} = (1-\theta)\left( \theta\BA\BA^{\top} + \frac{m}{d}\I_d\right), \\
 &  \sum_{j=1}^m\E \BDelta_{ij}^{\top} \BDelta_{ij} = \E \BDelta_i^{\top}\BDelta_i= \sum_{j=1}^m (1-\theta^2) \be_j\be_j^{\top} = (1-\theta^2)\I_m, \\
 & \sum_{i=1}^n \E\BDelta_i^{\top}\BDelta_i = \E \BDelta^{\top}\BDelta=  n(1-\theta^2)\I_m,
\end{align*}}
and then with $\theta\leq 1$, we have
\[
\max\left\{ \|\E \sum_{j=1}^m\E \BDelta_{ij} \BDelta_{ij}^{\top} \|, \| \E  \sum_{j=1}^m\E \BDelta_{ij}^{\top} \BDelta_{ij} \| \right\} \leq (1-\theta)\max\left\{ \left(\theta\|\BA\|^2 + \frac{m}{d}\right), 1+\theta\right\} \leq \frac{3(1-\theta)m}{d}.
\]
Note that $\|\BDelta_{ij}\| \leq 1 +\theta\leq 2$ 
and thus it holds
\begin{align*}
\max\|\BDelta_i\| 
& \leq \frac{4\gamma\log (n(m+d))}{3} + \sqrt{6\gamma (1-\theta)md^{-1}  \log (n(m+d))} \\
& \leq \left( \frac{8}{3} + 2\sqrt{3(1-\theta)\delta^{-1}} \right)  \log (n(m+d))
\end{align*}
with probability at least $1 - 2(n(m+d))^{-1}.$
Then the sufficient condition in Theorem~\ref{thm:gpm} becomes
\[
\frac{\theta}{2}\sqrt{\frac{m}{d}}\geq 70\kappa^4\sqrt{d}\max\left\{ \frac{8}{3},2 \sqrt{ 3(1-\theta)\delta^{-1}} \right\}\log (n(m+d))
\]
which is guaranteed by
\[
\theta \geq \frac{400\kappa^4 d\log (n(m+d))}{\sqrt{m}}~~~~\text{and}~~~~\theta \geq 1 - \frac{1}{(280\kappa^4\sqrt{3d\log (n(m+d))})^2}
\]
where $d\log (n(m+d)) = m\delta.$   
{



}


\end{document}